\documentclass[review]{elsarticle}

\usepackage{graphicx}
\usepackage{amsmath,amssymb,mathrsfs}

\newcommand{\Eye}{\hbox{1\kern-4truept 1}}

\newtheorem{theorem}{Theorem}

\newtheorem{corollary}[theorem]{Corollary}

\newtheorem{definition}[theorem]{Definition}
\newtheorem{example}[theorem]{Example}

\newtheorem{lemma}[theorem]{Lemma}

\newtheorem{proposition}[theorem]{Proposition}

\newenvironment{proof}[1][Proof]{\textbf{#1.} }{\ \rule{0.5em}{0.5em}}

\journal{Journal of \LaTeX\ Templates}

\begin{document}

\begin{frontmatter}

\title{Quantum Covariance and Filtering}

\author{John E. Gough\fnref{myfootnote}}
\address{Aberystwyth University, Wales, SY23 3BZ, UK}




\begin{abstract}
We give a tutorial exposition of the analogue of the filtering equation for quantum systems focusing on the quantum probabilistic framework and developing the ideas from the classical theory. Quantum covariances and conditional expectations on von Neumann algebras play an essential part in the presentation.
\end{abstract}


\end{frontmatter}

%\linenumbers

\section{Introduction}

Nonlinear filtering theory is a well-developed field of engineering which is used to estimate unknown quantities in the presence of noise. One of the founders of the field was the Soviet mathematician Ruslan Stratonovich who encouraged his student Viacheslav Belavkin to extend the problem to the quantum domain \cite{Belavkin1}. Classically, estimation works by measuring one or more variables which are dependent on the variables to estimated, and Bayes Theorem plays an essential role in inferring the unknown variables based on what we measure. 
Belavkin's approach uses the theory of quantum stochastic calculus for continuous-in-time homodyne and photon counting measurements. There are several approaches: in the paper of Barchielli and Belavkin \cite{Barchielli_Belavkin}, the characteristic functional method is used to derive the photon-counting case, with the diffusive case obtained as an appropriate limit. Further details of the many approaches and applications may be found in the books by Barchielli and Gregoratti \cite{Barchielli_Gregoratti} and by Wiseman and Milburn \cite{Wiseman_Milburn}.

However, the proof of Bayes Theorem requires a joint probability distribution for the unknown variables and the measured ones. Once we go to quantum theory, we have to be very careful as incompatible observables do not possess a joint probability distribution - in such cases, applying Bayes Theorem will lead to erroneous results and is the root of many of the paradoxes in the theory.

We will derive the simplest quantum filter. The filter equation itself was originally postulated by Gisin on different grounds of continuous collapse of the wavefunction, but subsequently given a standard filtering interpretation \cite{GG91}. It also appeared as way of simulating quantum open systems due to Carmichael \cite{Carmichael} and Dalibard, Castin and M\o lmer \cite{D}: while this appears as a trick for simulating just the quantum master equation (analogue of the Fokker-Planck equation) by stochastic processes, it is clear that the authors consider an underlying interpretation based on continual measurements. The discrete-time version of the filter also featured in the famous Paris Photon-Box experiment \cite{Paris_PB}.

%%%%%%%%%%%%%%%%%%%%%%%%%%%%%%%%%%%%%%%%%%%%%%%%%%%%%%%%%%%%%%%%%%%%%%%%%%%%%%%%%%%%%%%%%%%%%%%%%%%%%%%%%%%%%%

\section{Quantum Probabilistic Setting}
We start from the tradition formulation of quantum theory in terms of operators on a separable Hilbert space, $\mathfrak{h}$. The norm of a linear operator $X$ is $\| X \| = \sup \{ \| X \phi \| :
\phi \in \mathfrak{h} , \| \phi \| =1 \}$, and the collection of bounded operators will be denoted by $B( \mathfrak{h})$. We will denote the identity operator by $\Eye$. The adjoint of $X \in B(\mathfrak{h})$ will be denoted by $X^\ast$.

Our interest will be in von Neumann algebras. These are unital *-algebras with that are closed in the weak operator topology. Here we say that a sequence of operators $(X_n)$ converges weakly in $B(\mathfrak{h})$ to $X$ if their matrix elements converge, that is $\langle \phi , X_n \psi \rangle \to \langle \phi , X \psi \rangle$ for all $\phi, \psi \in \mathfrak{h}$.

A pair $(\mathfrak{A}, \langle \cdot \rangle )$ consisting of a von Neumann algebra and a state is referred to as a \textbf{quantum probability} (QP) space \cite{Maassen88}.  

\paragraph{Commutative = Classical}
Kolmogorov's setting for classical probability is in terms of probability spaces $(\Omega , \mathcal{A} , \mathbb{P})$ where $\Omega$ is a space of outcomes (the sample space), $\mathcal{A}$ is a $\sigma$-algebra of subsets of $\Omega$, and $\mathbb{P}$ is a probability measure on the elements in $\mathcal{A}$. The collection of functions $\mathfrak{A} = L^\infty (\Omega , \mathcal{A} , \mathbb{P})$ will form a commutative von Neumann algebra and, moreover, a state is given by $\langle A \rangle = \int_\Omega A(\omega ) \, \mathbb{P} [d \omega ]$. (Conversely, every commutative von Neumann algebra with a state that is continuous in the normal topology, see below, will be isomorphic to this framework.)

\paragraph{Commutants}
There is an alternative definition of von Neumann algebras which, surprising, is purely algebraic. For a subset of operators $\mathfrak{A}$, we define its commutant in $B (\mathfrak{h} )$ to be 
\begin{eqnarray}
\mathfrak{A}^\prime = \{ X \in B(\mathfrak{h }) : [A,X] =0 , \forall A \in \mathfrak{A} \}.
\end{eqnarray}
The commutant of the commutant of $\mathfrak{A}$ is called the bicommutant and is denoted $\mathfrak{A}^{\prime \prime}$. Von Neumann's Bicommutant Theorem states that a collection of operators $\mathfrak{A}$ is a von Neumann algebra if and only if it is closed under taking adjoints and $\mathfrak{A} = \mathfrak{A}^{\prime \prime}$.

$B( \mathfrak{h })$ itself is a von Neumann algebra. If $\mathfrak{A}$ and $\mathfrak{B}$ are von Neumann algebras then $\mathfrak{B}$ is said to be coarser than $\mathfrak{A}$ if $\mathfrak{B} \subset \mathfrak{A}$. A collection of operators $K$ generates a von Neumann algebra $ \mathrm{vN}(K)  = (K \cup K^\ast )^{\prime \prime})$. 

\paragraph{States} A state on a von Neumann algebra is a *-linear functional $\langle \cdot \rangle : \mathfrak{A} \mapsto \mathbb{C}$ which is positive ($\langle X \rangle \ge 0$ whenever $X \ge 0$) and normalized ($\langle \Eye \rangle =1$). We will assume that the state is continuous in the normal topology, that is $\sup_n \mathbb{E} [X_n] = \mathbb{E} [ \sup_n X_n ]$ for any increasing sequence $(X_n)$ of positive elements of $\mathfrak{A}$. The main point of interest is that the normal state takes the form $\langle X \rangle= \mathrm{tr} \{ \varrho X \}$ for $\varrho$ a density matrix.

The state satisfies the Cauchy-Schwartz identity $ | \langle X^\ast Y \rangle |^2 \le \langle X^\ast X \rangle \, \langle Y^\ast Y \rangle $.

\paragraph{Morphisms between QP Spaces}

A morphism $\phi : (\mathfrak{A}_1, \langle \cdot \rangle_1 ) \mapsto (\mathfrak{A}_2, \langle \cdot \rangle_2 )$ between QP spaces is a normal, completely positive, *-linear map which preserves the identity, $\phi (\Eye_1 ) = \Eye_2$, and the probabilities, $\langle \phi (X) \rangle_2 = \langle X \rangle _1$ for all $X \in \mathfrak{A}_1$. If a morphism is a homomorphism,  that is, $\phi (X) \phi (Y) = \phi (XY)$ for all $X,Y \in \mathfrak{A}_1$, then we say that $\mathfrak{A}_1$ is embedded into $\mathfrak{A}_2$.

\paragraph{Tomita-Takesaki Theory} As operators do not necessarily commute we may have $\langle X^\ast Y \rangle $ different from $\langle Y X^\ast \rangle$. Nevertheless, it is possible to write
\begin{eqnarray}
\langle Y X^\ast \rangle = \langle X^\ast \Delta Y \rangle ,
\end{eqnarray}
where $\Delta $ is a positive (possibly unbounded operator on $\mathfrak{A}$ known as the modular operator. This plays a central role in the Tomita-Takesaki theory of von Neumann algebras. A one-parameter group $\{ \sigma_t : t \in \mathbb{R} \}$ of maps on $\mathfrak{A}$ is defined by $\sigma_t (X) = \Delta^{-it} X \Delta^{it}$ and is known as the modular group associated with the QP space $(\mathfrak{A}, \langle \cdot \rangle )$.

\begin{theorem}[Takesaki, \cite{Takesaki72}]
Let $(\mathfrak{A}, \langle \cdot \rangle )$ be a QP space and let $\mathfrak{B}$ be a von Neumann subalgebra of $\mathfrak{A}$. There will exist a morphism $\mathfrak{E}$ from $ \mathfrak{A}$ down to $\mathfrak{B}$ which is projective ($ \mathfrak{E} \circ \mathfrak{E}= \mathfrak{E} $) if and only if $\mathfrak{B}$ is invariant under the modular group of $(\mathfrak{A}, \langle \cdot \rangle )$.
\end{theorem}

%%%%%%%%%%%%%%%%%%%%%%%%%%%%%%%%%%%%%%%%%%%%%%%%%%%%%%%%%%%%%%%%%%%%%%%%%%%%%%%%%%%%%%%%%%%%%%%%%%%%%%%%%%%%%%%%

\subsection{Quantum Conditioning}
We fix a QP space  $\big( \mathfrak{A}, \langle \cdot \rangle \big)$, and define the\textbf{ covariance} of two elements $X,Y \in \mathfrak{A}$ to be
\begin{eqnarray}
\mathrm{Cov} (X, Y ) \triangleq \langle X^\ast Y \rangle - \langle X \rangle ^\ast \langle Y \rangle .
\end{eqnarray}
Likewise the \textbf{variance} is defined as $\mathrm{Var} (X) \triangleq \mathrm{Cov} (X,X)$.

The idea is that we have a subset $\mathfrak{B} \subset \mathfrak{A}$, and we want to associate an element $\mathfrak{E} [A] \in \mathfrak{B}$ with each $A \in \mathfrak{A}$, see Figure \ref{fig:Venn_Neumann}.
As $\mathfrak{B}$ is smaller than $\mathfrak{A}$ we think of $\mathfrak{E}[A]$ as a coarse-grained version of $A$ based on a less information.
The map $\mathfrak{E}$ therefore compresses the model $(\mathfrak{A} , \langle \cdot \rangle )$ into a coarser one on $\mathfrak{B}$: we would like to do this is a way that preserves averages.

\begin{figure}[htbp]
	\centering
		\includegraphics[width=0.40\textwidth]{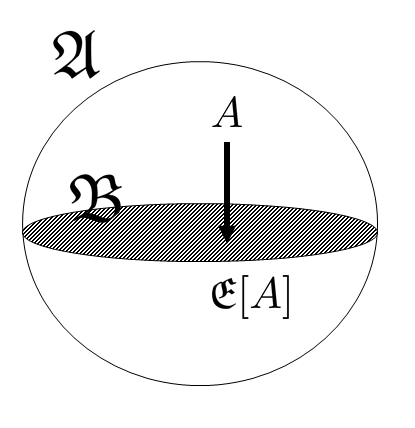}
	\caption{A conditional expectation $\mathfrak{E}$ is a projection from an algebra $\mathfrak{A}$ of random objects down into a smaller algebra $\mathfrak{B}$ such that $\langle \mathfrak{E} [A] \rangle  = \langle A \rangle$.}
	\label{fig:Venn_Neumann}
\end{figure}

We now list some desirable features for $\mathfrak{E}$ which we have already encountered in the classical case:
for any $X,Y,A \in \mathfrak{A}$, $\alpha , \beta \in \mathbb{C}$ and $B_1,B_2 \in \mathfrak{B}$,
\begin{enumerate}
\item[(CE1)] linearity: $\mathfrak{E} [ \alpha X + \beta Y ] = \alpha \mathfrak{E}[X] + \beta \mathfrak{E}[Y]$;
\item[(CE2)] *-map: $\mathfrak{E}[X^\ast  ] = \mathfrak{E}[X]^\ast$;
\item[(CE3)] conservativity: $\mathfrak{E}[\Eye ]= \Eye$;
\item[(CE4)] compatibility: $\langle \mathfrak{E}[A] \rangle = \langle A \rangle$;
\item[(CE5)] projectivity: $\mathfrak{E} [ \mathfrak{E}[A]] = \mathfrak{E}[A]$;
\item[(CE6)] peelability: $\mathfrak{E}[ B_1 A B_2 ] =B_1  \mathfrak{E}[  A ] B_2 $;
\item[(CE7)] positivity: $\mathfrak{E} [A] \ge 0$ whenever $A \ge 0$.
\end{enumerate}

We call property (CE6) \lq\lq peelability\rq\rq \, for the lack of a better name and we emphasize that the order of the operators is important. 
Property (CE7) is known to be insufficient to deal with quantum theory and must be strengthened as follows:
\begin{enumerate}
\item[(CE7$^\prime$)] complete positivity: for each integer $n \ge 1$
\begin{eqnarray}
\left[
\begin{array}{ccc}
\mathfrak{E} [A_{11}] & \cdots & \mathfrak{E}[A_{1n}] \\
\vdots & \ddots  & \vdots  \\
\mathfrak{E}[A_{n1}] & \cdots & \mathfrak{E}[A_{nn}] 
\end{array}
\right]
\ge 0
\; \mathrm{whenever} \,
\left[
\begin{array}{ccc}
A_{11} & \cdots & A_{1n} \\
\vdots & \ddots & \vdots \\
A_{n1} & \cdots & A_{nn} 
\end{array}
\right] \ge 0.
\end{eqnarray}
\end{enumerate}

\begin{definition}
Let $\mathfrak{A}$ and $\mathfrak{B}$ be a unital *-algebras with $\mathfrak{B}$ a subalgebra of $\mathfrak{A}$, then a mapping $\mathfrak{E}:\mathfrak{A} \mapsto \mathfrak{B}$ satisfying properties (CE1)-(CE6) and (CE7$^{\, \prime}$) is a quantum conditional expectation.
\end{definition}

\begin{proposition}
A quantum conditional expectation $\mathfrak{E}$ acts as the identity map on $\mathfrak{B}$.
\end{proposition}
\begin{proof}
Set $A=B_1 = \Eye$ and $B_2 = B\in \mathfrak{B}$, then peelability implies that $\mathfrak{E}[B] = \mathfrak{E}[\Eye ] B$. So the result follows from conservativity.
\end{proof}

\paragraph{Existence}
We observe that the conditional expectation always exists in the classical world. Here $\mathfrak{A}$ can be identified as some $ L^\infty (\Omega , \mathcal{A} , \mathbb{P})$ and then the subalgebra $\mathfrak{B}$ will be then take the form  $L^\infty (\Omega , \mathcal{B} , \mathbb{P})$ where $\mathcal{B}$ is a coarser $\sigma$-algebra. Conditional expectation is then well defined: For $A \in L^\infty (\Omega , \mathcal{A} , \mathbb{P})$ one sets $\mu_A [I] = \int_I A(\omega ) \mathbb{P} [d \omega ]$ for each $I \in \mathcal{B}$ then $\mu_A$ is absolutely continuous with respect to $\mathbb{P} |_{\mathcal{G}}$ and its Radon-Nikodym derivative is the conditional expectation which we denote as $\mathbb{E} [ A| \mathcal{B}]$. This is explicit in Kolmogorov's original paper.

In contrast, quantum conditional expectations need not exits. By definition, they satisfy the requirements of the Takesaki Theorem above (and additionally the peelability condition) so we need further invariance of the subalgebra $\mathfrak{B}$ under the modular group.

\subsection{Quantum Covariance}

\begin{definition}
Let $\mathfrak{E}$ be a quantum conditional expectation from $\mathfrak{A}$ onto a subalgebra $\mathfrak{B}$.
For each $A\in \mathfrak{A}$, we define $\delta A \triangleq A- \mathfrak{E}[A]$. The conditional covariance of $X,Y \in \mathfrak{A}$ is defined to be
\begin{eqnarray}
\mathrm{Cov}_{\mathfrak{B}} (X,Y) \triangleq \mathfrak{E} [ \delta X^\ast \, \delta Y ] .
\end{eqnarray}
The conditional variance is
\begin{eqnarray}
\mathrm{Var}_{\mathfrak{B}} (X) \triangleq \mathrm{Cov}_{\mathfrak{B}}  (X,X).
\end{eqnarray}
\end{definition}

Note that 
\begin{eqnarray}
\mathfrak{E} [\delta A] =
\langle \delta A \rangle =0 ,
\end{eqnarray}
for every $A \in \mathfrak{A}$. 
It is worth emphasizing that the conditional covariance defined here is an operator on $\mathfrak{B}$, not a scalar.

\begin{lemma}
\label{lem:delta}
We have $\mathfrak{E} [ B_1 \, \delta A \, B_2 ] =0$ whenever $A \in \mathfrak{A}$ and $B_1, B_2 \in \mathfrak{B}$. In particular, 
$\mathfrak{E} [ B \, \delta A  ] =0$ whenever $A \in \mathfrak{A}$ and $B \in \mathfrak{B}$. 
\end{lemma}

The proof depends crucially on peelability: $\mathfrak{E} [ B_1 \, \delta A \, B_2 ] = B_1 \,\mathfrak{E} [  \delta A ] \, B_2 =0$.
The following result is trivial classically, but again requires peelability in the non-commutative setting.

\begin{proposition}
The conditional covariance may alternatively be written as
\begin{eqnarray}
\mathrm{Cov}_{\mathfrak{B}} (X,Y) = \mathfrak{E} [X^\ast Y ] - \mathfrak{E} [X]^\ast \mathfrak{E}[Y] .
\label{eq:cov_form}
\end{eqnarray}
\end{proposition}
\begin{proof}
From \ref{lem:delta} we then have
\begin{eqnarray*}
\mathrm{Cov}_{\mathfrak{B}} (X,Y) &=& \mathfrak{E} \bigg[ X^\ast Y -\mathfrak{E} [X]^\ast Y - X^\ast \mathfrak{E}[Y] +  \mathfrak{E}[X]^\ast \mathfrak{E}[Y] \bigg] \\
 &=& \mathfrak{E}  [ X^\ast Y ] -\mathfrak{E} [X]^\ast \mathfrak{E}[Y] - \mathfrak{E}[X]^\ast \mathfrak{E}[Y] +  \mathfrak{E}[X]^\ast \mathfrak{E}[Y] 
\end{eqnarray*}
and the result follows.
\end{proof}

\begin{proposition}
The conditional covariance has the invariance property
\begin{eqnarray}
\mathrm{Cov}_{\mathfrak{B}} (X+B_1,Y+B_2) = \mathrm{Cov} _{\mathfrak{B}} (X,Y) ,
\label{eq:cov_invariance}
\end{eqnarray}
for all $B_1 , B_2 \in \mathfrak{B}$.
\end{proposition}
\begin{proof}
From *-linearity and (\ref{eq:cov_form}), we see that the left hand side of (\ref{eq:cov_invariance}) equals
\begin{eqnarray*}
 \mathfrak{E} \big[ X^\ast Y +X^\ast B_2 +B_1^\ast Y +B_1^\ast B_2  \big] 
-\big( \mathfrak{E}  [ X] +B_1 \big)^\ast \big( \mathfrak{E} [Y] +B_2 \big) 
\end{eqnarray*}
and the result follows using peelability.
\end{proof}

\begin{lemma}
\label{lem:cov}
The covariance and conditional covariance are related by
\begin{eqnarray}
\mathrm{Cov} (X,Y) = 
\langle \mathrm{Cov}_{\mathfrak{B}}  (X,Y) \rangle +
\big\langle ( \mathfrak{E}[X] - \langle X \rangle )^\ast
 ( \mathfrak{E}[Y] - \langle Y \rangle ) \big\rangle. 
\end{eqnarray}
\end{lemma}
\begin{proof}
This follows from repeated application of the compatibility property.
\begin{eqnarray*}
\langle \mathrm{Cov}_{\mathfrak{B}}  (X,Y) \rangle &=&
\langle X^\ast Y \rangle - \langle \mathfrak{E}[X]^\ast \mathfrak{E} [Y] \rangle \\
&=& \langle X^\ast Y \rangle - \langle X \rangle^\ast \langle Y \rangle 
- \big( \langle \mathfrak{E}[X]^\ast \mathfrak{E} [Y] \rangle
 - \langle X^\ast \rangle \langle Y \rangle ) \big) ,
\end{eqnarray*}
which is readily rearranged to give the result.
\end{proof}

As a consequence we have
\begin{eqnarray}
\mathrm{Var} (X) = 
\langle \mathrm{Var}_{\mathfrak{B}}  (X) \rangle +
\big\langle ( \mathfrak{E}[X] - \langle X \rangle )^\ast
 ( \mathfrak{E}[X] - \langle X \rangle ) \big\rangle. 
\end{eqnarray}

\subsection{Least Squares Property}

\begin{proposition}
\label{prop:least_squares_cond}
The conditional covariance has the least squares property, that is, $ \mathfrak{E} [ (X-B)^\ast (X-B) ]$ is minimized over $B \in \mathfrak{B}$ by $B=\mathfrak{E} [ X]$.
\end{proposition}
\begin{proof}
Let $B \in \mathfrak{B}$ then $B^\prime = B+ \mathfrak{E} [X]$ which is in again in $\mathfrak{B}$. Then
\begin{eqnarray*}
\mathfrak{E} [ (X-B)^\ast (X-B) ] &=& \mathfrak{E} [ (\delta X-B^\prime )^\ast (\delta X-B^\prime ) ]
\\
 &=& \mathfrak{E}  [\delta X^\ast \, \delta X  ] -B^{\prime \ast} \delta X - \delta X \, B^\prime 
+ B^{\prime \ast } B^\prime  ] \\
&=& \mathrm{Var} _{\mathfrak{B}} (X) + \mathfrak{E} [ B^{\prime \ast} B^\prime ] \\
&\ge & \mathrm{Var}_{\mathfrak{B}} (X) , 
\end{eqnarray*}
where we use the positivity property.
\end{proof}

\begin{corollary}
The variance $ \langle (X-B)^\ast (X-B) \rangle $ is also minimized over $B \in \mathfrak{B}$ by $B=\mathfrak{E} [ X]$.
\end{corollary}
\begin{proof}
Using the same notations from the proof of Lemma \ref{prop:least_squares_cond}, we have
\begin{eqnarray*}
\langle (X-B)^\ast (X-B) \rangle &=& \langle (\delta X-B^\prime )^\ast (\delta X-B^\prime ) \rangle
\\
 &=& \langle \delta X^\ast \delta X  \rangle - \langle B^{\prime \ast} \delta X \rangle - \langle \delta X^\ast \, B^\prime \rangle 
+ \langle  B^{\prime \ast } B^\prime  \rangle  \\
&=& \langle \delta X^\ast \delta X  \rangle +  \langle  B^{\prime \ast } B^\prime  \rangle , 
\end{eqnarray*}
since $\langle B^{\prime \ast} \delta X \rangle = \langle \delta X^\ast \, B^\prime \rangle =0$ by Lemma \ref{lem:delta}. Therefore $ \langle (X-B)^\ast (X-B) \rangle $ is also minimized over $B \in \mathfrak{B}$ by $B=\mathfrak{E} [ X]$.
\end{proof}
 
%%%%%%%%%%%%%%%%%%%%%%%%%%%%%%%%%%%%%%%%%%%%%%%%%%%%%%%%%%%%%%%%%%%%%%%%%%%%%%%%%%
\section{Classical Filtering}
In this section we recall in detail Kolmogorov’s Theory of Probability.  In the process we will see the commutative analogues that motivated the our more general definitions in the Introduction.

\subsection{Kolmogorov's Theory}
Kolmogorov's axiomatic formulation of probability theory is based on the mathematical formalism of measure theory. The main concept is that of a \textit{probability space}.  This is a triple $(\Omega, \mathcal{F} , \mathbb{P})$ where:

\begin{itemize}
\item $\Omega$, called the \textit{sample space}, is the collection of all possible outcomes (typically a topological space);
\item $\mathcal{F}$ is a $\sigma$-algebra of subsets of $\Omega$,the elements of which are known as \textit{events};
\item 
$\mathbb{P}$ is a \textit{probability measure} on $\mathcal{F}$.
\end{itemize}

In details, $\mathcal{F}$ will form a $\sigma$-algebra if it contains the empty set $\emptyset$, if it is closed under complementation (that is, if $A \in \mathcal{F}$ then so too will be its complement $A^\prime = \{ \omega \in \Omega: \omega \notin A \}$), and finally if whenever $\{ A_n \}$ is a countable number of events in $\mathcal{F}$ then their intersection $\cap_n A_n$ and union $\cup_n A_n$ will be in $\mathcal{F}$. Note that $\Omega $ will be an event since it is the complement of the empty set.

A probability measure $\mathbb{P}$ on $\mathcal{F}$ is an assignment of a probability $\mathbb{P} [A] \ge 0$ to each event $A \in \mathcal{F}$ with the rule that $\mathbb{P} [ \Omega] =1$ and $ \mathbb{P} [ \cap_n A_n ] = \sum_n \mathbb{P} [A_n]$ for any countable number of events, $\{ A_n \}$, that are non-overlapping (i.e., $A_n \cap A_m = \emptyset$ if $n \neq m$)

The pair $(\Omega , \mathcal{F} )$ comprise a \textit{measurable space}. In other words, a space where we are capable to assign possible measures of size to selected subsets in a consistent manner: this is the branch of mathematics known as \textit{measure theory} which was set up to resolve pathological problems when you try and assign a measure to all subsets. It follows that probability theory is formally just special case of measure theory where the measure $\mathbb{P}$ has maximum value $\mathbb{P} [\Omega]=1$. 

More exactly, the setting is measure theory but probability theory brings its own additionally concepts with it. An example is \textit{conditional probability}: the probability of event $A$ given that $B$ has occurred is defined by
\begin{eqnarray}
\mathbb{P} [A|B ] = \frac{  \mathbb{P} [ A \cap B ]}{ \mathbb{P} [B] }
\end{eqnarray}
which is the joint probability, $ \mathbb{P} [ A \cap B ]$, for both $A$ and $B$ to occur divided by the marginal probability $\mathbb{P} [B]$.

The choice of $\mathcal{F}$ in a given problem is part of the modeling process. Essentially, we have to ask what are the events that we want to assign a probability to. Let $\mathcal{G}$ be a $\sigma$-algebra that is contained in $\mathcal{F}$ (that is every event in $\mathcal{G}$  there is also an event in $\mathcal{F}$) then we say that $\mathcal{G}$ is \textit{coarser, or smaller}, than $\mathcal{F}$. The probability space $( \Omega , \mathcal{G} , \mathbb{Q})$ is then a \textit{coarse-graining} of $(\Omega , \mathcal{G} , \mathbb{P})$ where we take $\mathbb{Q}$ to be the restriction of $\mathbb{P}$ to the smaller $\sigma$-algebra $\mathcal{G}$.

Just as we do not consider all subsets of $\Omega$, we do not consider all functions on $\Omega$ either. Let $X: \Omega \to \mathbb{R}$ then we say $X$ is measurable with respect to a $\sigma$-algebra $\mathcal{F}$ if the sets 
\begin{eqnarray}
X^{-1} [I] \triangleq \{ \omega \in \Omega : X(\omega ) \in I \}
\end{eqnarray}
belong to $\mathcal{F}$ for each interval $I$. A measurable function $X$ on a probability space is called a \textit{random variable} and the probability that it takes a value in the interval $I$, denoted $\mathrm{Prob} \{ X \in I\} $ is just the value $\mathbb{P}$ assigns to the event $X ^ {-1} [I]$.
We will use the term \textit{random vector} for a vector-valued function whose components are all random variables.

Let $X_1 , \cdots , X_n$ be random variables, then there is a coarsest $\sigma$-algebra which contains all the events of the form $X_j^{-1}[I]$ for all $j$ and all intervals $I$: we refer to this as the $\sigma$-\textit{algebra generated by the random variables}.

The correct way to think of an ensemble is a probability space where $(\Omega ,\mathcal{F})$ is collection $\Gamma$ all possible microstates with $\mathcal{F}$ is some suitable $\sigma$-algebra, and $\mathbb{P}$ is a suitable probability measure. The Hamiltonian must, at the very least, be a measurable function with respect to whatever $\sigma$-algebra we chose. No philosophical interpretations needed beyond this point.

\subsection{Conditioning in Classical Probability}
 
We will now restrict attention to continuous random variables with well-defined probability densities.
A random variable $X$ has probability distribution function (pdf) $\rho _{X}$
so that 
\begin{eqnarray}
\Pr \left\{ x\leq X<x+dx\right\} =\rho _{X}\left( x\right) \,dx.
\end{eqnarray}
Normalization requires $\int_{-\infty }^{\infty }\rho _{X}\left( x\right)
dx=1$. If we have several random variables, then we need to specify their 
\textit{joint probability}. For instance, if we have a pair $X$ and $Y$ then
their joint pdf will be $\rho _{X,Y}\left( x,y\right) $ with 
\begin{eqnarray}
\rho _{X}\left( x\right) &=&\int \rho _{X,Y}\left( x,y\right) dy,\quad \mathrm{%
(}x\mathrm{-marginal)} \\
\rho _{Y}\left( y\right) &=&\int \rho _{X,Y}\left( x,y\right) dx,\quad \mathrm{%
(}y\mathrm{-marginal)}
\end{eqnarray}
and $ 1=\int \int \rho _{X,Y}\left( x,y\right) dxdy$.

We say that $X$ and $Y$ are \textit{statistically independent} if their
joint probability factors into the marginals 
\begin{eqnarray}
\rho _{X,Y}\left( x,y\right) \equiv \rho _{X}\left( x\right) \times \rho
_{Y}\left( y\right) ,\quad \text{( independence).}
\end{eqnarray}
This is equivalent to pairs of events of the form $X^{-1} [I]$ and $Y^{-1} [J]$ being statistically independent for all intervals $I,J$.

More generally, we can work out the conditional probabilities from a joint
probability. The pdf for $X$ given that $Y=y$ is defined to be 
\begin{eqnarray}
\rho _{X|Y}\left( x|y\right) \triangleq \frac{\rho _{X,Y}\left( x,y\right) }{%
\rho _{Y}\left( y\right) }.
\end{eqnarray}
In the special case where $X$ and $Y$ are independent we have 
\begin{eqnarray}
\rho _{X|Y}\left( x|y\right) =\rho _{X}\left( x\right) .
\end{eqnarray}
In other words, conditioning on the fact that $Y=y$ makes no change to our
knowledge of $X$.

\begin{definition}
Let $A = a (X,Y)$ be a random variable for some function $a: \mathbb{R}  \times \mathbb{R} \mapsto \mathbb{R}$, then its conditional expectation given $Y=y$ is defined to be
\begin{eqnarray}
\mathbb{E} [ A | Y=y  ] \triangleq \int_{\mathbb{R}} a (x , y ) \rho_{X|Y} (x | y) dx .
\end{eqnarray}
More generally, let $\mathcal{Y} $ be the $\sigma$-algebra generated by $Y$, then $\mathbb{E} [ A | \mathcal{Y}]$ is the $\mathcal{Y}$-measurable random variable taking the value $\mathbb{E} [ A | Y=y  ]$ for each $\omega$ where $y$ is the value of $Y (\omega )$.
\end{definition}

As $\int \rho_{X|Y} (x|y) \, dx =1$, we have
\begin{eqnarray}
\mathbb{E} [ 1 | \mathcal{Y}   ] \equiv 1 .
\end{eqnarray}
We note that for any random variable $A=a(X,Y)$ 
\begin{eqnarray}
\mathbb{E} \big[ \mathbb{E} [ A | \mathcal{Y} ] \big] &=& \int_{\mathbb{R}} \left(
\int_{\mathbb{R}} a (x , y )  \rho_{X|Y} (x | y) dx   \right) 
\rho_Y (x) \, dy
\nonumber \\
&=&
\int_{\mathbb{R}}dx
\int_{\mathbb{R}} d y \, a (x , y )  \rho_{X,Y} (x , y) dx   \nonumber \\
&=&
\mathbb{E}[A] .
\end{eqnarray}
Also, for any $A=a(X,Y)$ and $B=b(Y)$ we have
\begin{eqnarray}
 \mathbb{E} [ A B | \mathcal{Y} ] (\omega ) &=& 
\int_{\mathbb{R}} a (x , Y(\omega )  b(Y(\omega )) \rho_{X|Y} (x | y) dx   
\nonumber \\
&=&
\left(
\int_{\mathbb{R}}dx
\, a (x , Y(\omega ) )  \rho_{X|Y} (x | Y(\omega )) dx  \right) b (Y(\omega )) \nonumber \\
&=&
 \mathbb{E}[A| \mathcal{Y}] (\omega ) \, B (\omega ) .
\end{eqnarray}

This construction was specific to random variables with pdfs. However, it extends to the general setting as follows.

\begin{theorem}
Let $(\Omega , \mathcal{F} , \mathbb{P} )$ be a probability space and let $\mathcal{Y}$ be a sub-$\sigma$-algebra of $\mathcal{F}$. Then there exists a $\mathbb{P}$-almost surely unique  $\mathcal{Y}$-measurable random variable $\mathbb{E} [ X | \mathcal{Y} ]$ such that $\mathbb{E} [ 1 | \mathcal{Y}   ] =1$, $\mathbb{E} \big[ \mathbb{E} [ A | \mathcal{Y} ] \big] =\mathbb{E} [A]$ and $\mathbb{E} [ AB | \mathcal{Y} ] = \mathbb{E} [ A | \mathcal{Y} ] \, B$ whenever $B$ is $\mathcal{Y}$-measurable.
\end{theorem}

\begin{proposition}
If $B$ is $\mathcal{Y}$-measurable, then 
\begin{eqnarray}
\mathbb{E} [ B | \mathcal{Y} ] = B .
\end{eqnarray}
\end{proposition}
\begin{proof} Setting $A=1$ in the identity $\mathbb{E} [ AB | \mathcal{Y} ] = \mathbb{E} [ A | \mathcal{Y} ] \, B$ whenever $B$ is $\mathcal{Y}$-measurable, we see that $\mathbb{E} [ B | \mathcal{Y} ] =\mathbb{E} [ 1 | \mathcal{Y} ] B$ which in turn equals $B$. 
\end{proof}

\begin{proposition}
Conditional expectations are projections.
\end{proposition}
\begin{proof}  
For $A$ arbitrary, we set $B = \mathbb{E} [ A | \mathcal{Y} ]$ which is $\mathcal{Y}$-measurable and so
\begin{eqnarray}
\mathbb{E} [ \mathbb{E} [ A | \mathcal{Y} ] | \mathcal{Y} ] = \mathbb{E} [ A | \mathcal{Y} ].
\end{eqnarray}
\end{proof}

%%%%%%%%%%%%%%%%%%%%%%%%%%%%%%%%%%%%%%%%%%%%%%%%%%%%%%%%%%%%%%%%%%%
\subsection{Classical Measurement}
\label{sec:CM}

We now suppose that we have a system with phase space $\Gamma$ and a measuring apparatus with parameter space $M$. We let $x$ denote the phase points of $\Gamma$ as before, and write $y$ for the variables of the apparatus. The components of $y$ are sometimes referred to as \textit{pointer variables}. The total space will be $\Omega = \Gamma \times M$ with coordinates $\omega = (x,y)$. We take $\mathbb{P}$ to be a probability measure on $\Omega$ and consider the random vectors $X(\omega ) = x$ and $Y(\omega ) = y$.

In an experiment, we will not measure the system directly but instead record the value of one or more pointer variables. Let $\mathcal{Y}$ be the $\sigma$-algebra generated by $Y$. We therefore refer to $Y$ as \textit{the data}.

We shall assume that the system variables and the pointer variables are statistically dependent for our probability measure $\mathbb{P}$, otherwise we learn nothing about our system from the data. As before we assume a joint pdf $\rho (x,y)$ with marginals $\rho_\Gamma (x)$ for the system and $\rho_M (y)$ for the measuring apparatus. We will write $\rho (x|y)$ for the conditional pdf for our system given the data but write $\lambda (y |x)$ for the conditional pdf of the data given the system. This implies that
\begin{eqnarray}
\rho (x,y) = \rho( x | y ) \, \rho_M (y) = \lambda (y | x ) \, \rho_\Gamma (x) .
\end{eqnarray}

In practice, we may not know $\mathbb{P}$ however we will assume that we know $\lambda ( y | x)$. That is, we assume that we know the probability distribution of the pointer variables if we prepared our system precisely in state $x$, for each possible $x \in \Gamma$.
Statisticians refer to $\lambda \left( y|x\right) $ as the \textit{likelihood function} of the data $y$ given $x$.

Note that 
\begin{eqnarray}
\int _M \lambda \left( y|x\right) \, dy = \int_\Gamma \rho (y | x ) \, dx =1.
\end{eqnarray}

\bigskip

Now every random variable may be written as $A = a (X,Y)$ for some function $a: \Omega = \Gamma \times M \mapsto \mathbb{R}$. Its conditional expectation given the data is
\begin{eqnarray}
\mathbb{E} [ A | \mathcal{Y} ] \equiv\int_\Gamma a (x , Y ) \rho (x | Y) dx .
\end{eqnarray}
Indeed, for $\omega = (x,y)$ we have
\begin{eqnarray}
\mathbb{E} [ A | \mathcal{Y} ] (\omega ) &=& \frac{1}{ \rho_M (y)} \int_\Gamma a (x' , y ) \rho (x ', y) dx' \nonumber \\
                &=& \frac{  \int_\Gamma a (x' , y ) \rho (x ', y) dx'}{\int_\Gamma  \rho (x'', y) dx''} .
\label{eq:CEy}
\end{eqnarray}

This is an average over the hypersurface $\Omega_y =\{ \omega \in \Omega  : Y (\omega ) =y \}$. Indeed, the decomposition $\omega = (x,y)$ can be thought of as split into the constraint coordinates $y$ and the hypersurface coordinates $x$. 

\bigskip

From a practical stand point, we will have access only to the data - that is, variables measurable with respect to $\mathcal{Y}$ only. We are assuming that we know $\lambda $, which is the conditional probability for data given that the system. However, the problem is that the system is unknown and what we are given is, of course, the data. Therefore, we need to solve the inverse problem, namely to give the conditional probability for the unknown $X$ given the measured values for $Y$. The problem however is not well-posed. We do not have enough information in
the problem yet to write down the joint probability. 

To remedy this, we introduce a pdf for $X$ which is our \textit{a priori} guess: 
\begin{eqnarray}
\rho_X (x) \stackrel{\mathrm{guess!}}{=}
\rho _{\mathrm{prior}}\left( x\right) .
\end{eqnarray}
We then have the corresponding joint probability for $X$ and $Y$: 
\begin{eqnarray}
\rho _{\mathrm{prior}}\left( x,y\right) =\lambda \left( y|x\right) \times \rho
_{\mathrm{prior}}\left( x\right) .
\end{eqnarray}
If we subsequently measure $Y=y$ then we obtain the \textit{a posteriori} probability 
\begin{eqnarray}
\rho _{\mathrm{post}}\left( x|y\right)  &=&\frac{\rho _{X,Y}\left( x,y\right) 
}{\rho _{Y}\left( y\right) } \nonumber \\
&=&\frac{{\lambda \left( y|x\right) \rho _{\mathrm{prior}}\left( x\right) }}{%
\int \lambda \left( y|x^{\prime }\right) \rho _{\mathrm{prior}}\left(
x^{\prime }\right) dx^{\prime }}.
\end{eqnarray}

The conditional expectation in (\ref{eq:CEy}) can be then written as
\begin{eqnarray}
\mathbb{E} [ A | \mathcal{Y} ] (\omega )=
\frac{  
\int_\Gamma a (x' , y )  \lambda \left( y|x^{\prime }\right) \rho _{\mathrm{prior}}\left(
x^{\prime }\right) dx^{\prime }
}
{\int_\Gamma  \lambda \left( y|x^{\prime \prime}\right) \rho _{\mathrm{prior}}\left(
x^{\prime \prime}\right) dx^{\prime \prime }} .
\label{eq:CE_likelihood}
\end{eqnarray}

\begin{example}
\label{ex:s+n}
Let $X$ be the position of a particle. We measure 
\begin{eqnarray}
Y=X+\sigma Z
\end{eqnarray}
where $Z$ is a standard normal variable independent of $X$. We may refer to $X$ as the signal and $Z$ as the noise. 

Now if $X$ was known to be exactly $x$ then $Y$ will be normal with mean $x$ and variance $\sigma ^{2}$. Therefore, we can immediately write down the likelihood function: it is 
\begin{eqnarray}
\lambda \left( y|x\right) =\frac{1}{\sqrt{2\pi }\sigma }e^{-\left(
y-x\right) ^{2}/2\sigma ^{2}},
\end{eqnarray}
\begin{eqnarray}
\rho _{\mathrm{post}}\left( x|y\right) =\frac{\rho _{\mathrm{prior}}\left(
x\right) e^{-\left( y-x\right) ^{2}/2\sigma ^{2}}}{\int \rho _{\mathrm{prior}%
}\left( x^{\prime }\right) e^{-\left( y-x^{\prime }\right) ^{2}/2\sigma
^{2}}dx^{\prime }}.
\end{eqnarray}
In the special case where $X$ is assumed to be Gaussian, say mean $\mu _{0}$
and variance $\sigma _{0}^{2}$, we can give the explicit form of the
posterior as Gaussian with mean $\mu _{1}$ and variance $\sigma _{0}^{2}$
where 
\begin{eqnarray}
\mu _{1} &=&\frac{\sigma _{1}^{2}}{\sigma _{0}^{2}}\mu _{0}+\frac{\sigma
_{1}^{2}}{\sigma ^{2}}y \\
\frac{1}{\sigma _{1}^{2}} &=&\frac{1}{\sigma _{0}^{2}}+\frac{1}{\sigma ^{2}}.
\end{eqnarray}
There are two desirable features here. First, the new mean $\mu_1$ uses the data $y$. Second, the new variance $\sigma_1^2$ is smaller than the prior variance $\sigma_0^2$. In other words, the measurement is informative and decreases uncertainty in the state
\end{example}

\subsection{Classical Filtering}
It is possible to extend the conditioning problem to estimate the state of a dynamical system as it evolves in time based on continual monitoring. This involves the theory of stochastic processes and we will use the informal language of path integrals rather than the Ito calculus. 
\subsubsection{Stochastic Process}
A \textit{stochastic process} is a family, $\left\{ X\left( t\right) :t\geq
0\right\} $, of random variables labeled by time. The process is determined
by specifying all the multi-time distributions 
\begin{eqnarray}
\rho \left( x_{n},t_{n};\cdots ;x_{1},t_{1}\right)
\end{eqnarray}
for $X\left( t_{1}\right) =x_{1},\cdots ,X\left( t_{n}\right) =x_{n}$ for
each $n\geq 0$.

\bigskip

A stochastic process is said to be \textit{Markov} if the multi-time
distributions take the form 
\begin{eqnarray}
\rho \left( x_{n},t_{n};\cdots ;x_{1},t_{1}\right) = T ( x_n , t_n | x_{n-1}
, t_{n-1} ) \cdots T (x_2 , t_2 | x_1 , t_1 ) \, \rho (x_1 , t_1) ,
\end{eqnarray}
where whenever $t_n > t_{n-1} > \cdots > t_1$.

Here $T(x,t|x_{0},t_{0})$ is the probability density for $X(t)=x$ given that 
$X(t_{0})=x_{0}$, ($t>t_{0}$). 
\begin{eqnarray}
\text{Prob} \big\{
x \le X(t) \le x + dx | X(t_0) =x_0 \big\}
= T (x , t | x_0 , t_0 ) \, dx ,
\end{eqnarray}
for $t > t_0$.
It is called the \textit{transition mechanism}
of the Markov process. For consistency we should have the following propagation rule, known as the Chapman-Kolmogorov equation in probability theory,
\begin{eqnarray}
\int T( x, t | x_1 , t_1 )
 \, T (x_1 , t_1 | x_0 , t_0 ) \, dx_1
= T(x,t | x_0 , t_0 ) ,
\end{eqnarray}
for all $t > t_1 > t_0$.

\bigskip
\begin{example}
The Wiener process (Brownian motion) is determined by 
\begin{eqnarray}
T\left( x,t|x_{0},t_{0}\right) &=&\frac{1}{\sqrt{2\pi \left( t-t_{0}\right) }%
}e^{-\frac{\left( x-x_{0}\right) ^{2}}{2\left( t-t_{0}\right) }}, \\
\rho \left( x,0\right) &=&\delta _{0}\left( x\right) .
\end{eqnarray}
The transition mechanism here is the Green's function for the heat equation
\begin{eqnarray}
\frac{\partial}{\partial t} \rho = \frac{1}{2}
\frac{\partial^2}{\partial x^2} \rho .
\end{eqnarray}
(In other words, given the data $\rho ( \cdot, t_0) = f(\cdot )$ at time $t_0$, the solution for later times is $\rho (x,t) = \int T(x,t | x_0 , t_0 ) f (x_0) \, dx_0$.)

Norbert Wiener gave an explicit construction - known as the canonical version of Brownian motion, where the sample space is the space of continuous paths, $\mathbf{w}=\left\{ w\left(
t\right) :t\geq 0\right\} $, starting a the origin as sample space, with a suitable $\sigma$-algebra of subsets and
a well defined measure $\mathbb{P}_{\text{Wiener}}^{t}$.

The corresponding stochastic process is denote $W(t)$. Ito was able to construct a stochastic differential calculus around the Wiener process, and more generally diffusions, and we have the following Ito table
\begin{eqnarray}
\begin{tabular}{l|ll}
$\times $ & $dt$ & $dW$ \\ \hline
$dt$ & 0 & 0 \\ 
$dW$ & 0 &  $dt$ 
\end{tabular}.
\end{eqnarray}
\end{example}

\subsubsection{Path Integral Formulation}

Indeed, we have 
\begin{eqnarray}
\rho \left( x_{n},t_{n};\cdots ;x_{1},t_{1}\right) \,dx_{n}\cdots
dx_{1}\propto e^{-\sum_{k}\frac{\left( x_{k}-x_{k-1}\right) ^{2}}{2\left(
t_{k}-t_{k-1}\right) }}dx_{n}\cdots dx_{1}.
\end{eqnarray}
Formally, we may introduce a limit ``path integral'' with probability
measure on the space of paths 
\begin{eqnarray}
\mathbb{P}_{\text{Wiener}}^{t}\left[ d\mathbf{w}\right] =e^{-S_{\text{Wiener}%
}\left[ \mathbf{w}\right] }\mathcal{D}\mathbf{w}.
\end{eqnarray}
where we have the action 
\begin{eqnarray}
S_{\text{Wiener}}\left[ \mathbf{w}\right] =\int_{0}^{t}\frac{1}{2}\dot{w}%
\left( \tau \right) ^{2}d\tau .
\end{eqnarray}

For a diffusion $X\left( t\right) $ satisfying the Ito stochastic differential equation
\begin{eqnarray}
dX=v\left( X\right) dt+\sigma \left( X\right) dW
\end{eqnarray}
we have the corresponding measure 
\begin{eqnarray}
\mathbb{P}_{X}^{t}\left[ d\mathbf{x}\right] =e^{-S_{X}\left[ \mathbf{x}%
\right] }\mathcal{D}\mathbf{x}.
\end{eqnarray}
where we have the action (substitute $\dot{w}=\frac{\dot{x}-w}{\sigma }$
into $S_{\text{Wiener}}\left[ \mathbf{w}\right] $, and allow for a Jacobian
correction) 
\begin{eqnarray}
S_{X}\left[ \mathbf{x}\right] =\int_{0}^{t}\frac{1}{2}\frac{[\dot{x}%
-v(x)]^{2}}{\sigma (x)^{2}}d\tau +\frac{1}{2}\int_{0}^{t}\nabla .v(x)d\tau .
\end{eqnarray}

\subsubsection{The Classical Filtering Problem}

Suppose that we have a system described by a process $\left\{ X\left(
t\right) :t\geq 0\right\} $. We obtain information by observing a related
process $\left\{ Y\left( t\right) :t\geq 0\right\} $. 
\begin{eqnarray}
dX &=& v \left( X\right) dt+\sigma \left( X\right) dW\quad \text{(stochastic
dynamics),} \\
dY &=&h\left( X\right) dt+dZ\quad \text{(Noisy observations).}
\end{eqnarray}
Here we assume that the dynamical noise $W$ and the observational noise $Z$
are independent Wiener processes.

\bigskip

The joint probability of both $X$ and $Y$ up to time $t$ is 
\begin{eqnarray}
\mathbb{P}_{X,Y}^{t}\left[ d\mathbf{x},d\mathbf{y}\right] =e^{-S_{X,Y}\left[
x,y\right] }\mathcal{D}\mathbf{x}\mathcal{D}\mathbf{y},
\end{eqnarray}
where 
\begin{eqnarray}
S_{X,Y}\left[ \mathbf{x},\mathbf{y}\right]  &=&S_{X}\left[ \mathbf{x}\right]
+\int_{0}^{t}\frac{1}{2}\left[ \dot{y}-h\left( x\right) \right] ^{2}d\tau  \\
&=&S_{X}\left[ \mathbf{x}\right] +S_{\text{Wiener}}[\mathbf{y}]-\int_{0}^{t}%
\left[ h\left( x\right) \dot{y}-\frac{1}{2}h\left( x\right) ^{2}\right]
d\tau ,
\end{eqnarray}
or 
\begin{eqnarray}
\mathbb{P}_{X,Y}^{t}\left[ d\mathbf{x},d\mathbf{y}\right] =\mathbb{P}_{X}^{t}%
\left[ d\mathbf{x}\right] \mathbb{P}_{\mathrm{Wiener}}^{t}\left[ d\mathbf{y}%
\right] \,  \lambda \left( \mathbf{y}%
| \mathbf{x} \right) .
\end{eqnarray}
where the \textit{Kallianpur-Streibel likelihood}\footnote{Readers with a background in stochastic processes will recognize this as a Radon-Nikodym derivative associated with a Girsanov transformation.} is
\begin{eqnarray}
\lambda _t\left( \mathbf{y}|\mathbf{x} \right)  =e^{\int_{0}^{t}\left[ h\left( x\right) dy(\tau )-\frac{1}{2}h\left(
x\right) ^{2}d\tau \right] }.
\end{eqnarray}
The distribution for $X\left( t\right) $ given observations $\mathbf{y}%
=\left\{ y\left( \tau \right) :0\leq \tau \leq t\right\} $ is then
\begin{eqnarray}
\rho_t\left( x |\mathbf{y}\right) =\frac{\int_{ x(0)=x_0}^{ x(t) =x } \lambda _t\left( \mathbf{y}%
| \mathbf{x}\right) \mathbb{P}_{X}^{t} 
\left[ d\mathbf{x}\right] }{\int_{x(0) =x_0} \lambda_t \left( \mathbf{y}%
| \mathbf{x}^\prime \right) \mathbb{P}_{X}^{t}%
\left[ d\mathbf{x}^\prime \right]}
\end{eqnarray}

\bigskip

Let us write $\mathcal{Y}_t$ for the $\sigma$-algebra generated by the observations $ \{ Y(\tau ) : 0 \le \tau \leq t\} $. 
The estimate for $f (X(t))$ for any function $f$ conditioned on the observations up to time $t$ is called the \textit{filter} and, generalizing (\ref{eq:CE_likelihood}) to continuous time, we may write this as
\begin{eqnarray}
\mathfrak{E}_t (f) &=&  \mathbb{E} [ f(X(t)) | \mathcal{Y}_t ] \nonumber \\
&=& 
\int \rho _t \left( x |\mathbf{y}\right) f(x) \, dx =
 \frac{ \int \sigma_t (x|\mathbf{y}) f(x) dx }{\int \sigma_t (x^\prime|\mathbf{y})
dx^\prime }
\label{eq:filter_pi}
\end{eqnarray}
where  $\sigma_t (x|\mathbf{y}) =\int_{ x(0)=x_0}^{ x(t) =x } \lambda \left( \mathbf{y}%
| \mathbf{x}\right) \mathbb{P}_{X}^{t} \left[ d\mathbf{x}\right]$ is a non-normalized density. We introduce the stochastic process $\sigma_t (x) : \omega \mapsto \sigma _t (x | \mathbf{y} )$ and it can be shown to satisfy the \textit{Duncan-Mortensen-Zakai equation}
\begin{eqnarray}
d \sigma_t (x) = \mathcal{L}^\ast \sigma_t (x) \, dt
+h(x) \sigma_t (x) \, dY(t) .
\end{eqnarray}

This implies the \textit{filtering equation}
\begin{eqnarray}
d \mathfrak{E}_t (f) =\mathfrak{E}_t (\mathcal{L}  f) \, dt + \big\{ \mathfrak{E}_t (fh) - \mathfrak{E}_t (f) \mathfrak{E}_t (h) \big\} dI(t) ,
\end{eqnarray}
where the \textit{innovations process} is defined as
\begin{eqnarray}
dI(t)  = dY(t) - \mathfrak{E}_t (h ) \, dt . 
\end{eqnarray}

%%%%%%%%%%%%%%%%%%%%%%%%%%%%%%%%%%%%%%%%%%%%%%%%%%%
\section{Quantum Filtering}
We now consider the quantum analogue of filtering. See also \cite{BouvanHJam07}-\cite{Rouchon}.
\subsection{Quantum Measurement}
\paragraph{The Basic Concepts}
The Born interpretation of the wave function, $\psi (x)$, in quantum mechanics is that $|\psi (x)|^2$ gives the probability density of finding the particle at position $x$. More generally, in quantum theory, observables are represented by self-adjoint operators on a Hilbert space. The basic postulate of quantum theory is that the pure states of a system correespond to normalized the wave functions, $\Psi$, and we will follow Dirac and denote these as kets $|\Psi \rangle$. When we measure an observable, the physical value we record will be an eigenvalue. If the state is $| \Psi \rangle$ then the average value of the observable represented by $\hat A$ is $\langle \hat A \rangle = \langle \Psi |  \hat A | \Psi \rangle$.

Let us recall that a Hermitian operator $\hat{P}$ is called an \textit{orthogonal projection} if it
satisfies $\hat{P}^{2}=\hat{P}$. Then if we have a Hermitian operator $\hat{A}$ with a discrete set of eigenvalues, then there exists
a collection of orthogonal projections $\hat{P}_{a}$ labeled by the
eigenvalues $a$, satisfying $\hat{P}_{a}\hat{P}_{a^{\prime }}=0$ if $a\neq
a^{\prime }$ and $\sum_{a}\hat{P}_{a}=\hat{I}$, such that 
\begin{eqnarray}
\hat{A}=\sum_{a}a\,\hat{P}_{a}.
\end{eqnarray}
This is the spectral decomposition of $\hat A$. The operators $\hat{P}_{a}$ project onto $\mathcal{E}_{a}$ which is the 
\textit{eigenspace} of $\hat{A}$ for eigenvalue $a$. In other words, $%
\mathcal{E}_{a}$ is the space of all eigenvectors of $\hat{A}$ having
eigenvalue $a$. The eigenspaces are orthogonal, that is $\langle \psi |\phi
\rangle =0$ whenever $\psi $ and $\phi $ lie in different eigenspaces (this
is equivalent to $\hat{P}_{a}\hat{P}_{a^{\prime }}=0$ if $a\neq a^{\prime }$%
), and every vector $|\psi \rangle $ can be written as a superposition of
vectors $\sum_{a}|\psi _{a}\rangle $ where $|\psi _{a}\rangle $ lies in
eigenspace $\mathcal{E}_{a}$. (In fact, $|\psi _{a}\rangle =\hat{P}_{a}|\psi
\rangle $.)

We note that, for any integer $n$, 
\begin{eqnarray}
\hat{A}^{n}=\sum_{a}a^{n}\,\hat{P}_{a}
\end{eqnarray}
and any real $t$
\begin{eqnarray}
e^{it\hat{A}}=\sum_{a}e^{ita}\hat{P}_{a}.
\end{eqnarray}

Suppose we prepare a quantum system in a state $|\Psi \rangle $ and perform
a measurement of an observable $\hat{A}$. We know that we may only measure
an eigenvalue $a$ and quantum mechanics predicts the probability $p_{a}$. In
fact, using the spectral decomposition 
\begin{eqnarray}
\langle \hat{A}^n \rangle =\langle \sum_{a}a^n \,\hat{P}_{a}\rangle
=\sum_{a}\langle a^n\,\hat{P}_{a}\rangle =\sum_{a}a^n \,p_{a},
\end{eqnarray}
and so 
\begin{eqnarray}
p_{a}=\langle \hat{P}_{a}\rangle \equiv \langle \Psi |\hat{P}_{a}|\Psi
\rangle .
\end{eqnarray}

For the special case of a non-degenerate eigenvalue $a$, we have that the
eigenspace $\mathcal{E}_{a}$ is spanned by a single eigenvector $|a\rangle $%
, which we take to be normalized. In this case we have $\hat{P}%
_{a}=|a\rangle \langle a|$

\begin{eqnarray}
p_{a}=\langle \Psi |\hat{P}_{a}|\Psi \rangle =\langle \Psi |a\rangle \langle
a|\Psi \rangle \equiv \left| \langle a|\Psi \rangle \right| ^{2}.
\end{eqnarray}
We see that if an observable $\hat{A}$ has a non-degenerate eigenvalue $a$
with normalized eigenvector $|a\rangle $, then if the system is prepared in
state $|\Psi \rangle $, the probability of measuring $a$ in an experiment is 
$\left| \langle a|\Psi \rangle \right| ^{2}$. The modulus squared of an
overlap in this way may therefore have the interpretation as a probability.

The degenerate case needs some more attention. Here the eigenspace $\mathcal{E}_{a}$ can
spanned by a set of orthonormal vectors $|a1\rangle ,|a2\rangle ,\cdots $ so
that $\hat{P}_{a}=\sum_{n}|an\rangle \langle an|$, and so $%
p_{a}=\sum_{n}\left| \langle an|\Psi \rangle \right| ^{2}$. The choice of
the orthonormal basis for $\mathcal{E}_{a}$ is not important!

The probability $p_{a}$ is equal to the length-squared of $\hat{P}_{a}|\Psi
\rangle $, that is, 
\begin{eqnarray}
p_{a}= \| \hat{P}_{a}\Psi  \| ^{2}.
\end{eqnarray}
To see this, note that $ \| \hat{P}_{a}\Psi  \| ^{2}$ is the overlap of the
ket $\hat{P}_{a}|\Psi \rangle $ with its own bra $\langle \Psi |\hat{P}%
_{a}^{\dag }$ so 
\begin{eqnarray}
\| \hat{P}_{a}\Psi \| ^{2}=\langle \Psi |\hat{P}_{a}^{\dag }\,%
\hat{P}_{a}|\Psi \rangle =\langle \Psi |\hat{P}_{a}^{2}|\Psi \rangle
=\langle \Psi |\hat{P}_{a}|\Psi \rangle =p_{a}
\end{eqnarray}
where we used the fact that $\hat{P}_{a}=\hat{P}_{a}^{\dag }=\hat{P}_{a}^{2}$%
.

In the picture below, we project $|\Psi \rangle $ into the eigenspace $%
\mathcal{E}_{a}$ to get $\hat{P}_{a}|\Psi \rangle $. In the special case
where $|\Psi \rangle $ was already in the eigenspace, it equals its own
projection ($\hat{P}_{a}|\Psi \rangle =|\Psi \rangle $)\ and so $p_{a}=1$
since the state $|\Psi \rangle $ is normalized. If the state $|\Psi \rangle $
is however orthogonal to the eigenspace then its projection is zero ($\hat{P}%
_{a}|\Psi \rangle =0$) and so $p_{a}=0$.

In general, we get something in between. In the picture below we see that $%
|\Psi \rangle $ has a component in the eigenspace and a component orthogonal
to it. The projected vector $\hat{P}_{a}|\Psi \rangle $ will then have
length less than the original $|\Psi \rangle $, and so $p_{a}<1$.

\begin{figure}[tbph]
\centering
\includegraphics[width=0.750\textwidth]{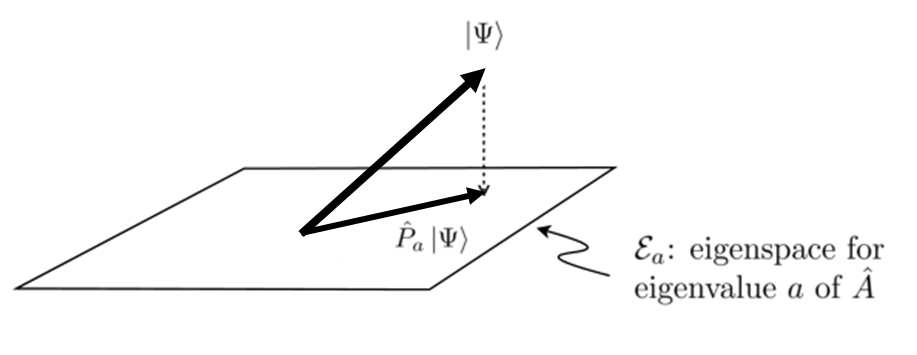} \label%
{fig:Projection_post}
\caption{The state $|\Psi \rangle $ is projected into the eigenspace $%
\mathcal{E}_{a}$ corresponding to the eigenvalue $a$ of $\hat{A}$.}
\end{figure}

\paragraph{Von Neumann's Projection Postulate}
Suppose the initial state is $|\Psi \rangle $ and we measure the eigenvalue $%
a$ of observable $\hat{A}$ in an given experiment. A second measurement of $%
\hat{A}$ performed straight way ought to yield the same value $a$ again,
this time with certainty.

The only way however to ensure that we measure a given eigenvalue with
certainty is if the state lies in the eigenspace for that eigenvalue. We
therefore require that the state of the system immediately after the result $%
a$ is measured will jump from $|\Psi \rangle $ to something lying in the
eigenspace $\mathcal{E}_{a}$. This leads us directly to the von Neumann
projection postulate.

\bigskip

\textbf{The von Neumann projection postulate:} 
\textit{If the state of a system is given by a ket $|\Psi \rangle $, and a
measurement of observable $\hat{A}$ 
yields the eigenvalue $a$, then the state immediately after measurement
becomes}
$ |\Psi _{a}\rangle =\dfrac{1}{\sqrt{p_{a}}}\,\hat{P}_{a}|\Psi
\rangle .$

We note that the projected vector $\hat{P}_{a}|\Psi \rangle $ has length $%
\sqrt{p_{a}}$ so we need to divide by this to ensure that $|\Psi _{a}\rangle 
$ is properly normalized. The von Neumann postulate is essentially the
simplest geometric way to get the vector $|\Psi \rangle $ into the
eigenspace: project down and then normalize!

\paragraph{Compatible Measurements}

Suppose we measure a pair of observables $\hat{A}$ and $\hat{B}$ in that
sequence. The $\hat{A}$-measurement leaves the state in the eigenspace of
the measured value $a$, the subsequent $\hat{B}$-measurement then leaves the
state in the eigenspace of the measured value $b$. If we then went back and
remeasured $\hat{A}$ would be find $a$ again with certainty? The state after
the second measurement will be an eigenvector of $\hat{B}$ with eigenvalue $b
$, but this need not necessarily be an eigenvector of $\hat{A}$.

Let $A$ and $\hat{B}$ be a pair of observables with spectral decompositions $%
\sum_{a}a\hat{P}_{a}$ and $\sum_{b}b\hat{Q}_{b}$ respectively. Let us
measure $\hat{A}$ and then $\hat{B}$ recording values $a$ and $b$
respectively. If the initial state was $|\Psi _{\text{in}}\rangle $ then we
obtain after both measurements the final state will be 
\begin{eqnarray}
|\Psi _{\text{out}}\rangle \propto \hat{Q}_{b}\hat{P}_{a}\,|\Psi _{\text{in}%
}\rangle .
\end{eqnarray}
In particular $|\Psi _{\text{out}}\rangle $ is an eigenstate of $\hat{B}$
with eigenvalue $b$. However suppose we also wanted $|\Psi _{\text{out}%
}\rangle $ to be an eigenstate of $\hat{A}$ with the original eigenvalue $a$%
, the we must have $\hat{P}_{a}|\Psi _{\text{out}}\rangle =|\Psi _{\text{out}%
}\rangle $ or equivalently 
\begin{eqnarray}
\hat{P}_{a}\hat{Q}_{b}\hat{P}_{a}\,|\Psi _{\text{in}}\rangle =\hat{Q}_{b}%
\hat{P}_{a}\,|\Psi _{\text{in}}\rangle .
\end{eqnarray}
If we want this to be true irrespective of the actual initial state $|\Psi _{%
\text{in}}\rangle $ then we arrive at the operator equation 
\begin{eqnarray}
\hat{P}_{a}\hat{Q}_{b}\hat{P}_{a}=\hat{Q}_{b}\hat{P}_{a}.
\end{eqnarray}

\begin{proposition}
Let $\hat{P}$ and $\hat{Q}$ be a pair of orthogonal projections satisfying $%
\hat{P}\hat{Q}\hat{P}=\hat{Q}\hat{P}$ then $\hat{P}\hat{Q}=\hat{Q}\hat{P}$.
\end{proposition}

\begin{proof}
We first observe that $\hat{R}=\hat{Q}\hat{P}\hat{Q}$ will again be an
orthogonal projection. To this end we must show that $R^{\dag }=R$ and $%
R^{2}=R$. However, $R^{\dag }=\left( \hat{Q}\hat{P}\hat{Q}\right) ^{\dag }=%
\hat{Q}^{\dag }\hat{P}^{\dag }\hat{Q}^{\dag }=\hat{Q}\hat{P}\hat{Q}=R$ and 
\begin{eqnarray*}
\hat{R}^{2} &=&\left( \hat{Q}\hat{P}\hat{Q}\right) \left( \hat{Q}\hat{P}\hat{%
Q}\right) =\hat{Q}\hat{P}\hat{Q}^{2}\hat{P}\hat{Q} \\
&=&\hat{Q}\hat{P}\hat{Q}\hat{P}\hat{Q}=\hat{Q}(\hat{P}\hat{Q}\hat{P})\hat{Q}
\\
&=&\hat{Q}(\hat{Q}\hat{P})\hat{Q}=\hat{Q}^{2}\hat{P}\hat{Q} \\
&=&\hat{Q}\hat{P}\hat{Q}=\hat{R}.
\end{eqnarray*}

However we also have $\hat{R}=\hat{Q}\hat{P}$, so the relation $\hat{R}=\hat{%
R}^{\dag }$ implies that $\hat{Q}\hat{P}=\hat{P}^{\dag }\hat{Q}^{\dag }=\hat{%
P}\hat{Q}$.
\end{proof}

We see that our operator identity above means that $\hat{Q}_{a}$ and $\hat{P}%
_{b}$ need to commute! If we wanted the $\hat{B}$-measurement not to disturb
the $\hat{A}$-measurement for any possible outcome $a$ and $b$, then we
require that all the eigen-projections of $\hat{A}$ commute with all the
eigen-projections of $\hat{B}$, and this implies that .

\begin{definition}
A collection of observables are compatible if they commute. We define the
commutator of two operators as 
\begin{eqnarray}
\left[ \hat{A},\hat{B}\right] =\hat{A}\hat{B}-\hat{B}\hat{A}
\end{eqnarray}
\end{definition}

So $\hat{A}$ and $\hat{B}$\ are compatible if $\left[ \hat{A},\hat{B}\right]
=0$.

\paragraph{Von Neumann's Model of Measurement}
The postulates of quantum mechanics outlined above assume that all measurements are idealized, but one might expect the actual process of extracting information from quantum systems to be more involved. Von Neumann modeled the measurement process as follows.
We wish to get information about an observable, $\hat{X}$, say the position of a quantum system. Rather than measure 
$\hat{X}$ directly, we measure an observable $\hat{Y}$ giving the pointer
position of a second system (called the measurement apparatus).

We will reformulate the von Neumann measurement problem in the language of estimation theory.
First we assume
that apparatus is described by a wave-function $\phi $. The initial state of
the system and apparatus is $|\Psi _{0}\rangle =|\Psi _{\mathrm{prior}}\rangle
\otimes |\phi \rangle $, i.e., 
\begin{eqnarray}
\langle x,y|\Psi _{0}\rangle =\Psi _{\mathrm{prior}}\left( x\right) \,\phi
\left( y\right) .
\end{eqnarray}
(Note  that we are already falling in line with the estimation way of thinking by referring to the initial wave function of the particle as an \textit{a priori} wave function - it is something we have to fix at the outset, even if we recognize it as only a guess for the correct physical state.))
The system and apparatus are taken to interact by means of the unitary 
\begin{eqnarray}
\hat{U}=e^{i\mu \hat{X}\otimes \hat{P}_{\mathrm{app}}/\hbar }
\end{eqnarray}
where $\hat{P}_{\mathrm{app}}=-i\hbar \frac{\partial }{\partial y}$ is the
momentum operator of the pointer conjugate to $\hat{Y}$. After coupling, the
joint state is 
\begin{eqnarray}
\langle x,y|\hat{U}\Psi _{0}\rangle =\Psi _{\mathrm{prior}}\left( x\right)
\,\phi \left( y-\mu x\right) .
\end{eqnarray}
If the measured value of $\hat{Y}$ is $y$, then the \textit{a posteriori wave-function} must be
\begin{eqnarray}
\psi _{\mathrm{post}}(x|y)=\frac{1}{\sqrt{\rho _{Y}(y)}}\psi _{\mathrm{prior}}      \left( x\right) \,\phi \left( y-\mu x\right) 
\end{eqnarray}
where 
\begin{eqnarray}
\rho _{Y}(y)=\int |\psi _{\mathrm{prior}}\left( x\right) \,\phi \left( y-\mu
x\right) |^{2}dx.
\end{eqnarray}
Basically, the pointer position will be a random variable with pdf given by $\rho _{Y}$: the \textit{a posteriori} wave-function may then be thought of as a random
wave-function on the system Hilbert space:
\begin{eqnarray}
\psi_{\mathrm{prior}} (x) \longrightarrow \psi _{\mathrm{post}}(x|Y).
\end{eqnarray}
In the parlance of quantum theorists, the wave function of the apparatus collapses to $| y \rangle$, while we update the \textit{a priori} wave function to get the \textit{a posteriori} one.

We have been describing events in the Schr\"{o}dinger picture where states evolve while observables remain fixed. In this picture, we measure the observable $\hat Y^{\mathrm{in}} =I \otimes \hat Y$, but the state is changing in time.  It is instructive to describe events in the Heisenberg picture. Here the state is fixed as $|\Psi _{0}\rangle =|\Psi _{\mathrm{prior}}\rangle
\otimes |\phi \rangle $, while the observables evolve. In fact, the observable that we actually measure is 
\begin{eqnarray}
\hat Y^{\text{out}} = \hat U^\ast \big( 
I \otimes \hat Y \big) \hat U = 
 \underbrace{\mu \, \hat U^\ast \big(  \hat X \otimes I \big)  \hat U}_{\mathrm{signal}} + \underbrace{\hat Y^{\mathrm{in}}}_{\mathrm{noise}},
\end{eqnarray}
from which it is clear that we are obtaining some information about $\hat X$. Note that the measured observable $ \hat Y^{\text{out}}$ is explicitly of the form \emph{signal} plus \emph{noise} as in Example \ref{ex:s+n}. The noise term, $\hat Y^{\text{in}}$, is independent of the signal and has the prescribed pdf  $ | \phi (y) |^2$.

\subsection{Quantum Markovian Systems}

\paragraph{Quantum Systems with Classical Noise}
We consider a quantum system driven by Wiener noise. For $H$ and $R$ self-adjoint, we set
\begin{eqnarray}
U(t) = e^{-iH t -iR W(t) } ,
\end{eqnarray}
which clearly defines a unitary process. From the Ito calculus we can quickly deduce the corresponding Schr\"{o}dinger equation
\begin{eqnarray}
dU(t) = \big[ -iH - \frac{1}{2} R^2 \big] U(t) \, dt
-iR U(t) \, dW(t) .
\end{eqnarray}
If we set $j_t (X) = U(t)^\ast XU(t)$, which we may think of as an embedding of the system observable $X$ into a noisy environment, then we similarly obtain
\begin{eqnarray}
dj_t (X) 
= j_t \big( \mathcal{L} (X) \big) \, dt
-i j_t \big(  [X,R] \big) \, dW(t) .
\end{eqnarray}
where
\begin{eqnarray}
\mathcal{L} (X) =
-i[X,H] - \frac{1}{2} \big[ [X,R],R \big] .
\end{eqnarray}

An alternative is to use Poissonian noise. Here we apply a unitary kick, $S$, at times distributed as a Poisson process with rate $\nu>0$. Let $N(t)$ count the number of kicks up to time $t$, then $\{ N(t) : t \ge 0 \}$ is a stochastic process with independent stationary increments (like the Wiener process) and we have the Ito rules
\begin{eqnarray}
dN(t) \, dN(t) = dN(t), \qquad \langle dN(t) \rangle
= \nu \, dt .
\end{eqnarray}
The Schr\"{o}dinger equation is $dU(t) = (S-I) U(t) \, dN(t)$ and for the evolution of observables we now have
\begin{eqnarray}
dj_t (X) = j_t \big( \mathcal{L} (X)\big) dN(t), \qquad
\mathcal{L}(X) = S^\ast X S -X.
\end{eqnarray}

\paragraph{Lindblad Generators}
A quantum dynamical semigroup is a family of CP maps, 
$\{ \Phi_t: t \geq 0\}$, such that $\Phi_t \circ \Phi_s
=\Phi_{t+s}$ and $\Phi (I) =I$. Under various continuity conditions one can show that the general form of the generator is
\begin{eqnarray}
\mathcal{L} (X) = \sum_k \frac{1}{2} L_k^\ast [X,L_k]
+ \sum_k \frac{1}{2} [L_k^\ast ,X] L_k - i [X,H ].
\end{eqnarray} 
These include the examples emerging from classical noise above - in fact, combinations of the Wiener and Poissonian cases give the general classical case. But the class of Lindblad generators is strictly larger that this, meaning that we need quantum noise! This is typically what we consider when modeling quantum optics situation.

\subsection{Quantum Noise Models}

\paragraph{Fock Space}
We recall how to model bosonic fields. We wish to describe a typical pure state $| \Psi \rangle$ of the field. If we look at the field we expect to see a certain number, $n$, of particles at locations $x_1 , x_2 , \cdots , x_n$ and to this situation we assign a complex number (the probability amplitude) $\psi_n (x_1 , x_2, \cdots x_n ) $. As the particles are indistinguishable bosons, the amplitude should be completely symmetric under interchange of particle identities. 

The field however can have an indefinite number of particles - that is, it can be written as a superposition of fixed number states. The general form of a pure state for the field will be
\begin{eqnarray}
| \Psi \rangle = \big( \psi_0 , \psi_1 , \psi_2, \psi_3 ,
\cdots \big).
\end{eqnarray}
Note that the case $n=0$ is included and is understood as the vacuum state. Here $\psi_0 $ is a complex number, with $ p_0 =| \psi_0 |^2$ giving the probability for finding no particles in the field.

The probability that we have exactly $n$ particles is
\begin{eqnarray}
p_{n}=\int \left| \psi _{n}\left( x_{1},x_{2},\cdots ,x_{n}\right) \right|
^{2}dx_{1}dx_{2}\cdots dx_{n},
\end{eqnarray}
and the normalization of the state is therefore $\sum_{n=0}^\infty p_n =1$.

In particular, we take the vacuum state to be
\begin{eqnarray}
| \Omega \rangle = \big( 1 , 0,0,0, \cdots \big) .
\end{eqnarray}

The Hilbert space spanned by such indefinite number of 
indistinguishable boson states is called \textit{Fock Space}.

A convenient spanning set is given by the exponential vectors 
\begin{eqnarray}
\langle x_{1},x_{2},\cdots ,x_{n}|\exp \left( \alpha \right) \rangle =\frac{1%
}{\sqrt{n!}}\alpha \left( x_{1}\right) \alpha \left( x_{2}\right) \cdots
\alpha \left( x_{n}\right) .
\end{eqnarray}
They are, in fact, over-complete and we have the inner products
\begin{eqnarray}
&&\langle \exp \left( \alpha \right) |\exp \left( \beta \right) \rangle \nonumber \\
&=&\sum_{n}  \frac{1}{n!} \int \alpha \left(
x_{1}\right) ^{\ast }\cdots \alpha \left( x_{n}\right) ^{\ast }\beta \left(
x_{1}\right) \cdots \beta \left( x_{n}\right) \,dx_{1}\cdots dx_{n} \nonumber  \\
&=&e^{\int \alpha \left( x\right) ^{\ast }\beta \left( x\right) dx} \nonumber \\
&=&e^{\langle \alpha |\beta \rangle }.
\end{eqnarray}
The exponential vectors, when normalized, give the analogues to the coherent states for a single mode.

We note that the vacuum is an example: $| \Omega \rangle = | \exp (0) \rangle$.

\paragraph{Quanta on a Wire}
We now take our space to be 1-dimensional - a wire. Let's parametrize the position on the wire by variable $\tau$, and denote by $\mathfrak{F}_{[s,t]}$ the Fock space over a segment of the wire $s \le \tau \le t$. We have the following tensor product decomposition \begin{eqnarray}
\mathfrak{F}_{A \cup B} = \mathfrak{F}_A \otimes
\mathfrak{F}_B, \qquad \qquad \text{if} \quad A \cap B = \emptyset
.
\end{eqnarray}

In is convenient to introduce quantum white noises $b(t)$ and $b(t)^\ast$ satisfying the singular commutation relations
\begin{eqnarray}
 [b(t) ,b(s)^\ast ] &=& \delta (t-s) .
\end{eqnarray}
Here $b(t)$ annihilates a quantum of the field at location $t$. In keeping with the usual theory of the quantized harmonic oscillator, we take it that $b(t)$ annihilates the vacuum: $b(t) \, | \Omega \rangle  = 0$. More generally, this implies that
\begin{eqnarray}
b(t) \, | \exp (\beta ) \rangle =
\beta(t)\, | \exp (\beta ) \rangle .
\label{eq:eigen_b}
\end{eqnarray}
The adjoint $b(t)^\ast$ creates a quantum at position $t$.

The quantum white noises are operator densities and are singular, but their integrated forms do correspond to well defined operators
which we call the \textit{annihilation and creation processes}, respectively,
\begin{eqnarray}
B(t) = \int_0^t b( \tau ) d \tau , \qquad
B(t)^\ast = \int_0^t b(\tau )^\ast d \tau .
\end{eqnarray}
We see that 
\begin{eqnarray}
[B(t) , B(s)^\ast ]
= \int_0^t d\tau \int_0^s d \sigma \, \delta (\tau - \sigma )
= \text{min} (t,s).
\end{eqnarray}

In addition we introduce a further process, called the \textit{number process}, according to
\begin{eqnarray}
\Lambda (t) = \int_0^t b(\tau )^\ast b ( \tau ) d \tau .
\end{eqnarray}

\paragraph{Quantum Stochastic Models}
We now think of our system as lying at the origin $ \tau =0$ of a quantum wire. The quanta move along the wire at the speed of light, $c$, and the parameter $\tau$ can be thought of as $x/c$ which is the time for quanta at a distance $x$ away to reach the system. Better still $\tau$ is the time at which this part of the field passes through the system. The process $B(t) = \int_0^t b( \tau ) d\tau$ is the operator describing the annihilation of quanta passing through the system at some stage over the time-interval $[0,t]$.

Fix a system Hilbert space, $\mathfrak{h}_0$, called the \textit{initial space}. A quantum stochastic process is a family of operators, $ \{ X(t): t \ge 0\}$, acting on $\mathfrak{h}_0 \otimes \mathfrak{F}_{[0, \infty )}$.                         .

The process is \textit{adapted} if, for each $t$, the operator $X(t)$ acts trivially on the future environment factor                .

QSDEs with adapted coefficients where originally introduced by Hudson \& Parthasarathy in 1984. Let $\{ X_{\alpha \beta } (t) : t \ge 0\}$ be four adapted quantum stochastic processes defined for $\alpha , \beta  \in \{ 0,1 \}$. We then define consider the QSDE
\begin{eqnarray}
\dot X (t) =
b(t)^\ast (t) X_{11} (t) b(t) +
b(t)^\ast X_{10} +
X_{01} (t) b(t)
+X_{00} (t),
\label{eq:fluxion_QSDE}
\end{eqnarray}
with initial condition $X(0) = X_0 \otimes I$. To understand this we take matrix elements between states of the form $| \phi \otimes \exp (\alpha ) \rangle$ and use the eigen-relation (\ref{eq:eigen_b}) to get the integrated form
\begin{eqnarray}
\langle \phi \otimes \exp ( \alpha ) |  X (t) | \psi \otimes \exp (\beta ) \rangle = \langle \phi    | X_0  | \psi \rangle \, \langle \exp ( \alpha ) |\exp (\beta ) \rangle 
\end{eqnarray}
\begin{eqnarray}
+&& \int_0^t \alpha(\tau )^\ast\langle \phi \otimes \exp ( \alpha ) |  X_{11} (t) | \psi \otimes \exp (\beta ) \rangle  \beta (\tau ) d\tau 
\nonumber \\
+&& \int_0^t \alpha(\tau )^\ast \langle \phi \otimes \exp ( \alpha ) |  X_{10} (t) | \psi \otimes \exp (\beta ) \rangle   d\tau 
\nonumber \\
+&& \int_0^t \langle \phi \otimes \exp ( \alpha ) |  X_{01} (t) | \psi \otimes \exp (\beta ) \rangle   \beta (\tau ) d\tau
\nonumber  \\
+&& \int_0^t \langle \phi \otimes \exp ( \alpha ) |  X_{00} (t) | \psi \otimes \exp (\beta ) \rangle     d\tau 
.
\end{eqnarray}
Processes obtain this way are called quantum stochastic integrals.

The approach of Hudson and Parthasarathy is actually different \cite{HP84,Par92}. The arrive at the process defined by (\ref{eq:fluxion_QSDE}) by building the analogue of the Ito theory for stochastic integration: that is the show conditions in which
\begin{eqnarray}
dX(t) &=& X_{11}(t) \otimes d\Lambda (t)
+X_{10} (t) \otimes dB(t)^\ast
+X_{01} (t) \otimes dB (t)
+ X_{00} (t) \otimes dt ,\nonumber  \\
\quad
\end{eqnarray}
makes sense as a limit process where all the increments are future pointing. That is $\Delta \Lambda \equiv \Lambda (t + \Delta t) - \Lambda (t)$ with $\Delta t >0$, etc.

One has, for instance,
\begin{eqnarray}
&&\langle \phi \otimes \exp ( \alpha ) |  X_{00} (t) \otimes \Delta B (t)| \psi \otimes \exp (\beta ) \rangle    \nonumber 
\\ && \quad = \bigg( \int_t^{t+\Delta t}   \beta (\tau ) d\tau \bigg) \times  
\langle \phi \otimes \exp ( \alpha ) |  X_{00} (t) \otimes I| \psi \otimes \exp (\beta ) \rangle ,
\end{eqnarray}
etc., so the two approaches coincide.

\paragraph{Quantum Ito Rules}
It is clear from (\ref{eq:fluxion_QSDE}) that this calculus is Wick ordered - note that the creators $b(t)^\ast$ all appear to the left and all the annihilators, $b(t)$, appear to the right of the coefficients. The product of two Wick ordered expressions in not immediately Wick ordered and one must use the singular commutation relations to achieve this. This results in a additional term which corresponds to a quantum Ito correction. 

We have 
\begin{eqnarray}
dB(t) dB(t) = dB(t)^\ast dB (t) = dB^\ast (t) dB^\ast (t)=0
\end{eqnarray}
To see this, let $X_t $ adapted, then
\begin{eqnarray}
\langle \exp (\alpha) | X_t dB(t)^\ast dB(t) | \exp (\beta) \rangle
= \alpha (t)^\ast \langle \exp ( \alpha ) | X_t \exp (\beta ) \rangle
\beta (t)
\,  (dt)^2
\end{eqnarray}
As we have a square of $dt$ we can neglect such terms. 

However, we have
\begin{eqnarray}
[ B(t) - B(s) , B (t)^\ast - B(s)^\ast ] = t-s,
\qquad (t>s)
\end{eqnarray}
and so $\Delta B \, \Delta B^\ast = \Delta B^\ast
\Delta B + \Delta t$. The infinitesimal form of this is then
\begin{eqnarray}
dB(t) dB(t)^\ast = dt .
\end{eqnarray}
This is strikingly similar to the classical rule for increments of the Wiener process!

In fact, we have the following quantum Ito table
\begin{eqnarray}
\begin{tabular}{l|llll}
$\times $ & $dt$ & $dB$ & $dB^{\ast }$ & $d\Lambda $ \\ \hline
$dt$ & 0 & 0 & 0 & 0 \\ 
$dB$ & 0 & 0 & $dt$ & $dB$ \\ 
$dB^{\ast }$ & 0 & 0 & 0 & 0 \\ 
$d\Lambda $ & 0 & 0 & $dB^{\ast }$ & $d\Lambda $%
\end{tabular}
.
\end{eqnarray}
Each of the non-zero terms arises from multiplying two processes that are not in Wick order.

For a pair of quantum stochastic integrals, we have the following quantum Ito product formula
\begin{eqnarray}
d \big( XY \big)
=(dX) dY + dX (dY) + (dX) (dY).
\end{eqnarray}
Unlike the classical version, the order of $X$ and $Y$ here is crucial.

\paragraph{Some Classical Processes On Fock Space}
The process $Q(t) = B(t) + B(t)^\ast $ is self-commuting, that is $[Q(t),Q(s) ] =0, \quad \forall t,s$, and has the distribution of a Wiener process is the vacuum state
\begin{eqnarray}
\langle \dot Q (t) \rangle &=& \langle \Omega
| [
b(t) + b(t)^\ast ] \Omega \rangle =0, \\
\langle \dot Q(t) \dot Q(s) \rangle
&=& \langle \Omega | b(t) b^\ast (s) \Omega \rangle
= \delta (t-s) .
\end{eqnarray}

The same applies to $P(t) = \frac{1}{i} [ B(t) - B(t)^\ast ]$, but
\begin{eqnarray}
[Q(t), P(s)] =2i \, \text{min}(t,s).
\end{eqnarray}
So we have two non-commuting Wiener processes in Fock space. We refer to $Q$ and $P$ as canonically conjugate quadrature processes.

One see that, for instance,
\begin{eqnarray}
dQ dQ = dB dB^\ast = dt.
\end{eqnarray}

We also obtain a Poisson process by the prescription
\begin{eqnarray}
N(t) = \Lambda (t) + \sqrt{\nu} B^\ast (t)
+\sqrt{\nu} B(t) + \nu t.
\end{eqnarray}
One readily checks that $dN dN = dN$ from the quantum Ito table.

\paragraph{Emission-Absorption Interactions}
Let us consider a singular Hamiltonian of the form
\begin{eqnarray}
\Upsilon (t) =
H \otimes I
+i
L \otimes b(t)^\ast - i L^\ast \otimes b(t).
\label{eq:Upsilon}
\end{eqnarray}
We will try and realize the solution to the Schr\"{o}dinger equation
\begin{eqnarray}
\dot U (t) = -i \Upsilon (t) \, U(t), \qquad U(0)=I.
\label{eq:QSDE_flux_Up}
\end{eqnarray}
as a unitary quantum stochastic integral process.

Let us first remark that the annihilator part of (\ref{eq:Upsilon}) will appear out of Wick order when we consider (\ref{eq:QSDE_flux_Up}). The standard approach in quantum field theory is to develop the unitary $U(t)$ as a Dyson series expansion - often re-interpreted as a time order-exponential:
\begin{eqnarray}
U(t) &=& I -i \int_0^t \Upsilon (\tau ) U(\tau ) 
d \tau \nonumber \\
&=&  1 - i \int_0^t d\tau \Upsilon (\tau )
+ (-i)^2 \int_0^t d \tau_2 \int_0^{\tau_2} d\tau_2
\Upsilon (\tau _2 ) \Upsilon (\tau _1 ) + \cdots \nonumber  \\
&=& \vec{T} e^{-i \int_0^t \Upsilon ( \tau ) d\tau} .
\end{eqnarray}
In our case the field terms - the quantum white noises - are linear, however, we have the problem that they come multiplied by the system operators $L$ and $L^\ast$ which do not commute, and don't necessarily commute with $H$ either.

Fortunately we can do the Wick ordering in one fell swoop rather than having to go down each term of the Dyson series.
We have
\begin{eqnarray}
\left[ b\left( t\right) ,U\left( t\right) \right]  &=&\left[ b\left(
t\right) ,I-i\int_{0}^{t}\Upsilon \left( \tau \right) U\left( \tau \right)
d\tau \right] =-i\int_{0}^{t}\left[ b\left( t\right) ,\Upsilon \left( \tau
\right) \right] U\left( \tau \right) d\tau  \nonumber  \\
&=&\int_{0}^{t}\left[ b\left( t\right) ,Lb\left( \tau \right) ^{\ast }\right]
U\left( \tau \right) d\tau \nonumber  \\
&=&L\int_{0}^{t}\delta \left( t-\tau \right)
U\left( \tau \right) d\tau =\frac{1}{2}LU\left( t\right) ,
\end{eqnarray}
where we dropped the $[b(t) , U(\tau )]$ term as this should vanish for $t> \tau$ and took half the weight of the $\delta$-function due to the upper limit $t$ of the integration. However, we get
\begin{eqnarray}
b\left( t\right) U\left( t\right) =U\left( t\right) b\left( t\right) +\frac{1%
}{2}LU\left( t\right) .
\end{eqnarray}
Plugging this into the equation (\ref{eq:QSDE_flux_Up}), we get
\begin{eqnarray}
\dot{U}\left( t\right)  &=&b\left( t\right) ^{\ast }LU\left( t\right)
-L^{\ast }b\left( t\right) U\left( t\right) -iH\left( t\right) U\left(
t\right) \nonumber  \\
&=&b\left( t\right) ^{\ast }LU\left( t\right) -L^{\ast }U\left( t\right)
b\left( t\right) -\left( \frac{1}{2}L^{\ast }L+iH\right) U\left( t\right) .
\end{eqnarray}
which is now Wick ordered. We can interpret this as the Hudson-Parthasarathy equation
\begin{eqnarray}
dU\left( t\right) =\left\{ L\otimes dB\left( t\right) ^{\ast }-L^{\ast
}\otimes dB\left( t\right) -\left( \frac{1}{2}L^{\ast }L+iH\right) \otimes
dt\right\} U\left( t\right) .
\end{eqnarray}

The corresponding Heisenberg equation for $j_t (X) = U(t)^\ast [X \otimes I ] U(t) $ will be

\begin{eqnarray}
dj_{t}\left( X\right)  &=&dU\left( t\right) ^{\ast }\left[ X\otimes I\right]
U\left( t\right) +U\left( t\right) ^{\ast }\left[ X\otimes I\right] dU\left(
t\right) \nonumber  \\
&& +dU\left( t\right) ^{\ast }\left[ X\otimes I\right] dU\left(
t\right) \nonumber  \\
&=&j_{t}\left( \mathcal{L}X\right) \otimes dt+j_{t}\left( \left[ X,L\right]
\right) \otimes dB\left( t\right) ^{\ast }+j_{t}\left( \left[ L^{\ast },X%
\right] \right) \otimes dB\left( t\right) 
\end{eqnarray}
where
\begin{eqnarray}
\mathcal{L}X &=&-X\left( \frac{1}{2}L^{\ast }L+iH\right) -\left( \frac{1}{2}%
L^{\ast }L-iH\right) X+L^{\ast }XL \nonumber \\
&=&\frac{1}{2}\left[ L^{\ast },X\right] L+\frac{1}{2}L^{\ast }\left[ X,L%
\right] -i\left[ X,H\right] .
\end{eqnarray}
We note that we obtain the typical Lindblad form for the generator.

\paragraph{Scattering Interactions}
We mention that we could also treat a Hamiltonian with only scattering terms
Let us set $\Upsilon \left( t\right) =E\otimes b\left( t\right) ^{\ast }b\left(
t\right) $. The same sort of argument leads to
\begin{eqnarray}
\left[ b\left( t\right) ,U\left( t\right) \right] =-iE\int_{0}^{t}\left[
b\left( t\right) ,b\left( \tau \right) ^{\ast }\right] b\left( \tau \right)
U\left( \tau \right) d\tau =-\frac{i}{2}Eb\left( t\right) U\left( t\right) ,
\end{eqnarray}
which can be rearranged to give
\begin{eqnarray}
b\left( t\right) U\left( t\right) =\frac{1}{I-\frac{i}{2}E}U\left( t\right)
b\left( t\right) .
\end{eqnarray}
So the Wick ordered form is
\begin{eqnarray}
\dot{U}\left( t\right) =Eb\left( t\right) ^{\ast }b\left( t\right) U\left(
t\right) =\frac{E}{I-\frac{i}{2}}b\left( t\right) ^{\ast }U\left( t\right)
b\left( t\right) 
\end{eqnarray}
or in quantum Ito form
\begin{eqnarray}
dU\left( t\right) =\left( S-I\right) \otimes d\Lambda \left( t\right)
\,U\left( t\right) ,\qquad \left( S=\frac{I+\frac{i}{2}E}{I-\frac{i}{2}E}%
\text{, unitary!}\right) .
\end{eqnarray}
The Heisenberg equation here is $dj_{t}\left( X\right) =j_{t}\left( S^{\ast
}XS-X\right) \otimes d\Lambda \left( t\right) $.

This is all comparable to the classical Poisson process driven evolution involving unitary kicks.

\paragraph{The SLH Formalism}
We now outline the so-called \textit{SLH} Formalism - named after the scattering matrix operator $S$, the coupling vector operator $L$ and Hamiltonian $H$ appearing in these Markov models \cite{GouJam09a}-\cite{CKS}.
The examples considered up to now used only one species of quanta. We could in fact have $n$ channels, based on $n$ quantum white noises:
\begin{eqnarray}
 [  b_j (t) ,
b^\ast_k (s) ] = \delta_{jk} \, \delta (t-s) .
\end{eqnarray}

The most general form of a unitary process with fixed coefficients may be described as follows:
we have a \textit{Hamiltonian} $H=H^\ast$, a column vector of coupling/ collapse operators
\begin{eqnarray}
 L=\left[
\begin{array}{c}
 L_{1} \\
\vdots  \\
L_{n}
\end{array}
\right] ,
\end{eqnarray}
and a matrix of operators
\begin{eqnarray}
 S=\left[
\begin{array}{ccc}
S_{11}  &\cdots  &  S_{1n}  \\
\vdots  & \ddots  & \vdots  \\
 S_{n1}  & \cdots  &  S_{nn} 
\end{array}
\right] ,
\qquad   S^{-1} =  S^\ast .
\end{eqnarray}

For each such triple $(S,L,H)$ we have the QSDE
\begin{eqnarray}
 d U(t) &=& \bigg\{
\sum_{jk}(  S_{jk} - \delta_{jk} I)\otimes
d \Lambda_{jk} (t)
+ \sum_j L_j  \otimes dB_j^\ast (t) \nonumber  \\
&&
- \sum_{jk} L_j^\ast S_{jk} \otimes dB_k (t) 
 -(\frac{1}{2} \sum_k L_k^\ast L_k +
i H ) \otimes dt \bigg\} \, U(t)
\label{eq:SLH_QSDE}
\end{eqnarray}
which has, for initial condition $U(0) =I$, a solution which is a unitary adapted quantum stochastic process. The emission-absorption case is the $n=1$ model with no scattering ($S=I$). Likewise the purse scattering corresponds to $H=0$ and $L=0$.

\paragraph{Heisenberg-Langevin Dynamics}
System observables evolve according to the Heisenberg-Langevin equation
\begin{eqnarray}
d j_t ( X) &=& \sum_{jk}
j_t 
(S^\ast_{lj}XS_{lk}- \delta_{jk} X)
d \Lambda_{jk} (t)
+ \sum_{jl}j_t ( S_{lj}^\ast [L_l,X])\otimes dB_j (t)^\ast \nonumber  \\
&&  + \sum_{lk} j_t ([X,L^\ast_l ] S_{lk}) \otimes dB_k (t)
+j_t ( \mathscr{L} X)\otimes dt .
\end{eqnarray}
where the generator is the traditional Lindblad form
\begin{eqnarray}
\mathscr{L} X = \frac{1}{2}\sum_k L^\ast_k [X,L_k]
+ \frac{1}{2} \sum_k [L^\ast_k  , X] L_k -i [X,H ] .
\end{eqnarray}

\paragraph{Quantum Outputs}
The output fields are defined by
\begin{eqnarray}
B^{\text{out}}_k (t) =
U(t)^\ast [ I \otimes B_k (t) ] U(t).
\end{eqnarray}

From the quantum Ito calculus we find that
\begin{eqnarray}
dB^{\text{out}}_j (t) = \sum_k j_t (S_{jk}) \otimes dB_k (t)
+ j_t (L_k) \otimes dt ,
\end{eqnarray}
Or, maybe more suggestively in quantum white noise language \cite{Gardiner_Collett},
\begin{eqnarray}
b^{\text{out}}_j (t) = \sum_j j_t (S_{jk})
 \otimes b_k(t) +j_t (L_j) \otimes I.
\end{eqnarray}

\subsection{Quantum Filtering}
We now set up the quantum filtering problem. For simplicity, we will take $n=1$ and set $S=I$ so that we have a simple emission-absorption interaction. We will also consider the situation where we measure the $Q$-quadrature of the output.

The initial state is taken to be $|\psi_0 \rangle \otimes |\Omega \rangle $, and in the Heisenberg picture this is fixed for all time.

The analogue of the stochastic dynamical equation considered in the classical filtering problem is the Heisenberg-Langevin equation
\begin{eqnarray}
dj_{t}\left( X\right)  =j_{t}\left( \mathcal{L}X\right) \otimes dt+j_{t}\left( \left[ X,L\right]
\right) \otimes dB\left( t\right) ^{\ast }+j_{t}\left( \left[ L^{\ast },X%
\right] \right) \otimes dB\left( t\right) 
\end{eqnarray}
where $\mathcal{L}X =\frac{1}{2}\left[ L^{\ast },X\right] L+\frac{1}{2}L^{\ast }\left[ X,L%
\right] -i\left[ X,H\right] $.

Some care is needed in specifying what exactly we measure: we should really work in the Heisenberg picture for clarity. The $Q$-quadrature of the input field is 
$Q\left( t\right) =B\left( t\right) +B\left( t\right) ^{\ast }$ which we have already seen is a Wiener process for the vacuum state of the field. Of course this is not what we measure - we measure the output quadrature!

Set
\begin{eqnarray}
Y^{\text{in}}\left( t\right) =I\otimes Q\left( t\right) .
\end{eqnarray}
As indicated in our discussion on von Neumann's measurement model, what we actually measure is
\begin{eqnarray}
Y^{ \text{out}} (t) = U(t)^\ast Y^{\text{in}} (t) U(t) = B^{\text{out} } (t) + B^{\text{out}} (t)^\ast .
\end{eqnarray}
The differential form of this is
\begin{eqnarray}
dY^{\text{out}} (t) = dY^{\text{in}} (t) + j_t (L+L^\ast ) dt .
\end{eqnarray}
Note that
\begin{eqnarray}
dY^{\text{in}}\left( t\right) dY^{\text{in}}\left( t\right) =dt=dY^{\text{out%
}}\left( t\right) dY^{\text{out}}\left( t\right) .
\end{eqnarray}

The dynamical noise is generally a quantum noise and can only be considered classical in very special circumstances, while the observational noise is just its $Q$-quadrature which can hardly be treated as independent!

In complete contrast to the classical filtering problem we considered earlier, we have no paths for the system - just evolving observables of the system. What is more these observables do not typically commute amongst themselves, or indeed the measured process.

We can only apply Bayes Theorem in the situation where the quantities involved have a joint probability distribution, and in the quantum world this requires them to be compatible. At this stage it may seem like a miracle that we have any theory of filtering in the quantum world. However, let us stake stock of what we have.

\paragraph{What Commutes With What?}

For fixed $s \ge 0$, let $U(t,s)$ be the solution to the QSDE (\ref{eq:SLH_QSDE}) in time variable $ t \ge s$ with $U(s,s)=I$. Formally, we have 
\begin{eqnarray}
U\left( t,s\right) =\vec{T}e^{-i\int_{s}^{t}\Upsilon \left( \tau \right) d\tau }
\end{eqnarray}
which is the unitary which couples the system to the part of the field that enters over the time $s\leq
\tau \leq t$. In terms of our previous definition, we have $U(t) = U(t,0)$ and we have the property
\begin{eqnarray}
U\left( t\right) =U\left( t,s\right) U\left( s\right) ,\qquad \left(
t>s>0\right) .
\end{eqnarray}

In the Heisenberg picture, the observables evolve
\begin{eqnarray}
j_{t}\left( X\right)  &=&U\left( t\right) ^{\ast }\left[ X\otimes I\right]
U\left( t\right) , \\
Y^{\text{out}}\left( t\right)  &=&U\left( t\right) ^{\ast }\left[ I\otimes
Q\left( t\right) \right] U\left( t\right) .
\end{eqnarray}

We know that the input quadrature is self-commuting, but what about the output one?
A key identity here is that 
\begin{eqnarray}
Y^{\text{out}}\left( t\right) =U\left( t\right) ^{\ast }Y^{\text{in}}\left(
s\right) U\left( t\right) ,\qquad \left( t>s\right) ,
\end{eqnarray}
which follows from the fact that $\left[ Y^{\text{in}}\left( s\right)
,U\left( t,s\right) \right] =0$. 

\bigskip 

From this, we see that the process $Y^{\text{out}}$ is also commutative since
\begin{eqnarray}
\left[ Y^{\text{out}}\left( t\right) ,Y^{\text{out}}\left( s\right) \right]
=U\left( t\right) ^{\ast }\left[ Y^{\text{in}}\left( t\right) ,Y^{\text{in}%
}\left( s\right) \right] U\left( t\right) =0,\quad \left( t>s\right) .
\end{eqnarray}
If this was not the case then subsequent measurements of the process $Y^{\text{out}}$ would invalidate (disturb?) earlier ones. In fancier parlance, we say that process is \textit{not self-demolishing} - that is, all parts are compatible with each other.

A similar line of argument shows that 
\begin{eqnarray}
\left[ j_{t}\left( X\right) ,Y^{\text{out}}\left( s\right) \right] =U\left(
t\right) ^{\ast }\left[ X\otimes I,I\otimes Q\left( t\right) \right] U\left(
t\right) =0,\quad \left( t>s\right) .
\end{eqnarray}
Therefore, we have a joint probability for $j_{t}\left( X\right) $ and the continuous collection of observables $\left\{
Y^{\text{out}}\left( \tau \right) :0\leq \tau \leq t\right\} $ so can use
Bayes Theorem to estimate $j_t (X)$ for any $X$ using the past observations.
Following V.P. Belavkin, we refer to this as the \textit{non-demolition principle}.

\paragraph{The Conditioned State}
In the Schr\"{o}dinger picture, the state at time $t \ge 0$ is $|\Psi _{t}\rangle =U\left( t\right) |\phi \otimes
\Omega \rangle $, so
\begin{eqnarray}
d|\Psi _{t}\rangle  &=&-\left( \frac{1}{2}L^{\ast }L+iH\right) |\Psi
_{t}\rangle dt+LdB\left( t\right) ^{\ast }|\Psi _{t}\rangle -L^{\ast
}dB\left( t\right) |\Psi _{t}\rangle \nonumber  \\
&=&-\left( \frac{1}{2}L^{\ast }L+iH\right) |\Psi _{t}\rangle dt+LdB\left(
t\right) ^{\ast }|\Psi _{t}\rangle +LdB\left( t\right) |\Psi _{t}\rangle \nonumber   \\
&=&-\left( \frac{1}{2}L^{\ast }L+iH\right) |\Psi _{t}\rangle dt+LdY^{\text{in}} (t)|\Psi
_{t}\rangle .
\end{eqnarray}

Here we have used a profound trick due to A.S. Holevo. The differential $dB(t)$ acting on $ | \Psi_t \rangle$ yields zero since it is future pointing and so only affects the future part which, by adaptedness, is the vacuum state of the future part of the field. To get from the first line to the second line, we remove and add a term that is technically zero. In its reconstituted form, we obtain the $Q$-quadrature of the input. The result is that we obtain an expression for the state $| \Psi_t \rangle$ which is \lq\lq diagonal\rq\rq \, in the input quadrature - our terminology here is poor (we are talking about a state not and observable!) but hopefully wakes up physicists to see what's going on.

The above equation is equivalent to the SDE in the system Hilbert space
\begin{eqnarray}
d|\chi _{t}\rangle =-\left( \frac{1}{2}L^{\ast }L+iH\right) |\chi
_{t}\rangle dt+L|\chi _{t}\rangle dy_{t}
\label{eq:BZ}
\end{eqnarray}
where $\mathbf{y}$ is a sample path - or better still, \textit{eigen-path} - of the quantum stochastic process $Y^{\text{in}}$.

We refer to (\ref{eq:BZ}) as the \textit{Belavkin-Zakai equation}.

\paragraph{The Quantum Filter}
Let us begin with a useful computational
\begin{eqnarray}
\langle \phi \otimes \Omega |j_{t}\left( X\right) F\left[ Y_{\left[ 0,t%
\right] }^{\text{out}}\right] |\phi \otimes \Omega \rangle 
&=&
\langle \phi \otimes \Omega | U(t)^\ast \big(  X\otimes  F\left[ Y_{\left[ 0,t%
\right] }^{\text{in}}\right] \big) U(t) |\phi \otimes \Omega \rangle \notag \\
&=&
\langle \Psi_t |   X\otimes  F\left[ Y_{\left[ 0,t%
\right] }^{\text{in}}\right]  |\Psi_t \rangle \notag \\
 &=&
\int
\langle \chi_t (\mathbf{y} ) |  X\otimes |\chi_t (\mathbf{y}) \rangle \, F\left[ \mathbf{y}\right]  
\, \mathbb{P}_{\text{Wiener}} [d \mathbf{y}]. \notag  \\
&& \quad
\label{eq:big}
\end{eqnarray}

A few comments are in order here. The operator $j_{t}\left( X\right) $ will commute with any functional of the past measurements - here $F\left[ Y_{\left[ 0,t\right] }^{\text{out}}\right] $. In the first equality is pulling things back in terms of the unitary $U(t)$. The second is just the equivalence between Schr\"{o}dinger and Heisenberg pictures. The final one just uses the equivalent form (\ref{eq:BZ}): note that the paths of the input quadrature gets their correct weighting as Wiener processes.

Setting $X=I$ in (\ref{eq:big}), we get the
\begin{eqnarray}
\langle \phi \otimes \Omega | F\left[ Y_{\left[ 0,t%
\right] }^{\text{out}}\right] |\phi \otimes \Omega \rangle 
&=&
\int
\langle \chi_t (\mathbf{y} |  \chi_t (\mathbf{y} )\rangle \, F\left[ \mathbf{y}\right]  
\, \mathbb{P}_{\text{Wiener}} [d \mathbf{y}]
\end{eqnarray}
So the probability of the measured paths is
\begin{eqnarray}
\mathbb{Q} [d \mathbf{y}] =
\langle \chi_t (\mathbf{y} )|  \chi_t (\mathbf{y}) \rangle \, \mathbb{P}_{\text{Wiener}} [d \mathbf{y}] .
\end{eqnarray}
Now this last equation deserves some comment! The vector $| \Psi_t \rangle$, which lives in the system tensor Fock space, is properly normalized, but its corresponding form $ | \chi _t \rangle$ is not! The latter is a stochastic process taking values in the system Hilbert space and is adapted to input quadrature. However, we never said that $| \chi_t \rangle$ had to be normalized too, and indeed it follows from or \lq\lq diagonalization\rq\rq \, procedure. In fact, if $| \chi_t \rangle$ was normalized then the output measure would follow a Wiener distribution and  so we would be measuring white noise!

From (\ref{eq:big}) again, we an deduce the filter: we get (using the arbitrariness of the functional $F$)
\begin{eqnarray}
\mathfrak{E}_t (X) = \frac{ \langle \chi_t (\mathbf{y}) | X | \chi_t (\mathbf{y}) \rangle}{
\langle \chi_t (\mathbf{y} )|  \chi_t (\mathbf{y}) \rangle } .
\end{eqnarray}

This has a remarkable similarity to (\ref{eq:filter_pi}). Moreover, using the Ito calculus see that
\begin{eqnarray}
d \langle \chi_t (\mathbf{y}) | X | \chi_t (\mathbf{y}) \rangle
&=&
\langle \chi_t (\mathbf{y}) | \mathcal{L} X | \chi_t (\mathbf{y}) \rangle  dt \notag \\
&&+\langle \chi_t (\mathbf{y}) |  \big( XL+L^\ast X \big) |\chi_t (\mathbf{y} )\rangle \, dy(t).
\end{eqnarray}
This is the quantum analogue of the Duncan-Mortensen-Zakai equation.

So small work is left in order to derive the filter equation. We first observe that the normalization (set $X=I$) is that
\begin{eqnarray}
d \langle \chi_t (\mathbf{y}) |  \chi_t (\mathbf{y}) \rangle
=\langle \chi_t (\mathbf{y}) |  \big( L + L^\ast  \big) |\chi_t (\mathbf{y} )\rangle \, dy(t).
\end{eqnarray}
Using the Ito calculus, it is then routine to show that the quantum filter is
\begin{eqnarray}
 d \mathfrak{E}_t (X) =  \mathfrak{E}_t ( \mathcal{L}X) \, dt + \big\{ \mathfrak{E}_t (XL+L^\ast X ) - \mathfrak{E}_t (X) \mathfrak{E}_t (L+L^\ast ) \big\} dI(t)
\end{eqnarray}
where the innovations are defined by
\begin{eqnarray}
 dI (t) = dY^{\text{out}} (t) - \mathfrak{E}_t (L+L^\ast ) \, dt .
\end{eqnarray}
Again, the innovations have the statistics of a Wiener process. As in the classical case, the innovations give the difference between what we observe next, $dY^{\text{out}} (t)$, and what we would have expected based on our observations up to that point, $\mathfrak{E}_t (L+L^\ast ) \, dt $. The fact that the innovations are a Wiener process is a reflection of the efficiency of the filter - after extracting as much information as we can out of the observations, we are left with just white noise.

\section*{Acknowledgements}
I would like to thank the staff at CIRM, Luminy (Marseille), and at Institut Henri Poincar\'{e} (Paris) for their kind support during the 2018 Trimester on Measurement and Control of Quantum Systems where this work was began. I am also grateful to the other organizers Pierre Rouchon and Denis Bernard for valuable comments during the writing of these notes.

%\bibliography{mybibfile}

\begin{thebibliography}{99}




\bibitem{Belavkin1}  V.P. Belavkin, (1989),
Non-Demolition Measurements, Nonlinear Filtering and Dynamic Programming of Quantum Stochastic Processes, 
Lecture Notes in Control and Inform Sciences \textbf{121} 245--265, Springer--Verlag, Berlin.



\bibitem{Barchielli_Belavkin}
A. Barchielli and V. P. Belavkin, (1991), Measurements continuous in time and a posteriori states in quantum mechanics, 
J. Phys. A: Math. Gen. 24, 1495.

\bibitem{Barchielli_Gregoratti}
A. Barchielli and M. Gregoratti, (2009), Quantum Trajectories and Measurements in Continuous Time - the diffusive case, Springer Berlin Heidelberg.


\bibitem{Wiseman_Milburn}
H.M. Wiseman and G.J. Milburn, (2009), Quantum Measurement and Control, Cambridge University Press.



\bibitem{GG91}
D. Gatarek, N. Gisin, (1995),
Continuous quantum jumps and infinite‐dimensional stochastic equations,
Journal of Mathematical Physics 32 (8), 2152-2157

\bibitem{Carmichael}
H.J. Carmichael, (1993), Phys. Rev. Lett. 70(15) p.2273.

\bibitem{D}
J. Dalibard, Y. Castin, Yvan, and K. M\o lmer, (Feb 1992), Wave-function approach to dissipative processes in quantum optics. Phys. Rev. Lett. American Physical Society. 68 (5): 580–58.


\bibitem{Paris_PB} 
C. Sayrin, I. Dotsenko, et al., (1 September 2011), Real-time quantum feedback prepares and stabilizes photon number states,
 Nature 477, 73-77.


\bibitem{Maassen88}
H. Maassen, (1988), Theoretical concepts in quantum probability: quantum Markov processes. 
Fractals, quasicrystals, chaos, knots and algebraic quantum mechanics (Maratea, 1987), 287-302, 
NATO Adv. Sci. Inst. Ser. C Math. Phys. Sci., 235, Kluwer Acad. Publ., Dordrecht

\bibitem{Takesaki72} M. Takesaki, (1972) Conditional Expectations in von Neumann Algebras, J. Func. Anal., \textbf{9}, 306-321.

\bibitem{BouvanHJam07}  L. Bouten, R. van Handel and M.R. James,  (2007),
An introduction to quantum filtering,
SIAM Journal on Control and Optimization \textbf{46}, 2199.

\bibitem{BvH_ref}
L. Bouten, R. van Handel, Quantum filtering: a reference probability approach, aXiv:math-ph/0508006 

\bibitem{vH_thesis}
R. van Handel, Ph.D. Thesis, Filtering, Stability, and Robustness, CalTech, 2006,
http://www.princeton.edu/~rvan/thesisf070108.pdf


\bibitem{Wiseman}
H. Wiseman, (1994), Quantum theory of continuous feedback, Phys. Rev. A,
49(3):2133-2150.

\bibitem{Rouchon}
P. Rouchon, (August 13 - 21, 2014), Models and Feedback Stabilization of Open Quantum Systems
Extended version of the paper attached to an invited conference for the International Congress of Mathematicians in Seoul, 
arXiv:1407.7810 


\bibitem{HP84}  R.L. Hudson and K.R. Parthasarathy,  (1984), Quantum Ito's formula and
stochastic evolutions, Commun. Math. Phys. \textbf{93}, 301.


\bibitem{Par92}  K.R. Parthasarathy, (1992) 
\emph{An Introduction to Quantum Stochastic Calculus}, Birkhauser.



\bibitem{GouJam09a}  J. Gough, M.R. James,  (2009),
 Quantum Feedback Networks: Hamiltonian Formulation, 
Commun. Math. Phys. \textbf{287}, 1109.


\bibitem{GouJam09b}  J. Gough, M.R. James,  (2009),
The series product and its application to quantum feedforward and feedback networks,
IEEE Trans. on Automatic Control \textbf{54}, 2530.

\bibitem{CKS} J. Combes, J. Kerckhoff, M. Sarovar,  (2017), The SLH framework for modeling quantum input-output networks,
 Advances in Physics: X, 2:3, 784-888.


\bibitem{Gardiner_Collett}  C.W. Gardiner and M.J. Collett, (1985),  Input and output
in damped quantum systems: Quantum stochastic differential equations and the
master equation. Phys. Rev. A, \textbf{31}(6):3761-3774.


\end{thebibliography}

\end{document}